\documentclass[12pt]{article}
\usepackage[english]{babel}
\usepackage{amssymb,amsmath,graphicx}
\usepackage{amsthm}
\usepackage{float}
\usepackage{afterpage}
\usepackage{dcolumn}
\usepackage{multirow}
\usepackage{color}
\usepackage[round,authoryear]{natbib}
\usepackage{enumitem}
\usepackage{slashbox}
\usepackage{pict2e}
 \usepackage{setspace}
\usepackage{microtype}

\usepackage[noend]{algorithmic}
\usepackage{algorithm}
\usepackage{rotating}

\newtheorem{proposition}{Proposition}

\newcolumntype{d}[1]{D{.}{.}{#1}}

\newcommand{\cA}{{\cal A}}
\renewcommand{\Re}{{\mathbb{R}}}
\newcommand{\cG}{{\cal G}}
\newcommand{\cV}{{\cal V}}

\newcommand{\cE}{{\cal E}}
\newcommand{\cH}{{\cal H}}
\newcommand{\cS}{{\cal S}}

\newcommand{\cR}{{\mathcal R}}

\DeclareGraphicsExtensions{.png}
\newcolumntype{P}[1]{>{\centering}p{#1}}
\newcolumntype{M}[1]{>{\centering}m{#1}}
\usepackage{color}
\newcommand{\blue}[1]{#1}
\usepackage[table]{xcolor}
\usepackage{float}
\usepackage{color} 
\usepackage{multirow}
\usepackage{graphicx}

\usepackage{flafter}

\def\sqw{\hbox{\rlap{\leavevmode\raise.3ex\hbox{$\sqcap$}}$%
\sqcup$}}
\def\sqb{\hbox{\hskip5pt\vrule width4pt height6pt depth1.5pt%
\hskip1pt}}

\def\qed{\ifmmode\hbox{\hfill\sqb}\else{\ifhmode\unskip\fi%
\nobreak\hfil
\penalty50\hskip1em\null\nobreak\hfil\sqb
\parfillskip=0pt\finalhyphendemerits=0\endgraf}\fi}
\def\cqfd{\ifmmode\sqw\else{\ifhmode\unskip\fi\nobreak\hfil
\penalty50\hskip1em\null\nobreak\hfil\sqw
\parfillskip=0pt\finalhyphendemerits=0\endgraf}\fi}

\usepackage{natbib}

 \bibpunct[, ]{(}{)}{,}{a}{}{,}%
 \def\newblock{\ }%

\usepackage{algorithmic}
\usepackage{algorithm}

\DeclareGraphicsExtensions{.png}

\newcolumntype{M}[1]{>{\flushleft}m{#1}}

\title{Large neighborhoods with implicit customer selection for vehicle routing problems with profits}

\author{Thibaut Vidal, Nelson Maculan, Puca Huachi Vaz Penna, Luis Satoru Ochi}

\begin{document}

\begin{center}

\begin{LARGE}
Large neighborhoods with implicit customer selection for vehicle routing problems with profits
\end{LARGE}

\vspace*{0.65cm}

\textbf{Thibaut Vidal *} \\
LIDS, Massachusetts Institute of Technology \\
POM, University of Vienna, Vienna, Austria \\
vidalt@mit.edu \\
\vspace*{0.2cm}
\textbf{Nelson Maculan} \\
Universidade Federal do Rio de Janeiro (COPPE-IM-UFRJ), Brazil  \\
maculan@cos.ufrj.br \\
\vspace*{0.2cm}
\textbf{Luiz Satoru Ochi, Puca Huachi Vaz Penna} \\
Instituto de Computa\c{c}\~ao -- Universidade Federal Fluminense, Niter\'oi, Brazil \\
\{satoru,ppenna\}@ic.uff.br \\

\vspace*{0.7cm}

\begin{large}
Working Paper, MIT -- Revised, July 2014
\end{large}

\vspace*{0.4cm}

\end{center}
\noindent
\textbf{Abstract.}
We consider several Vehicle Routing Problems (VRP) with profits, which seek to select a subset of customers, each one being associated with a profit, and to design service itineraries. When the sum of profits is maximized under distance constraints, the problem is usually called team orienteering problem.
The capacitated profitable tour problem seeks to maximize profits minus travel costs under capacity constraints. Finally, in the VRP with private fleet and common carrier, some customers can be delegated to an external carrier subject to a cost. Three families of combined decisions must be taken: customers selection, assignment to vehicles, and sequencing of deliveries for each route.

We propose a new neighborhood search for these problems which explores an exponential number of solutions in pseudo polynomial time. \blue{The search is conducted with standard VRP neighborhoods on an \emph{exhaustive} solution representation, visiting all customers. Since visiting all customers is usually infeasible or sub-optimal,} an efficient \emph{Select} algorithm, based on resource constrained shortest paths, \blue{is repeatedly used on any new route to find the optimal subsequence of visits to customers.}
The good performance of these neighborhood structures is demonstrated by extensive computational experiments with a local search, an iterated local search and a hybrid genetic algorithm. Intriguingly, even a local-improvement method to the first local optimum of this neighborhood achieves an average gap of 0.09\% on classic team orienteering benchmark instances, rivaling with the current state-of-the-art metaheuristics. Promising research avenues on hybridizations with more standard routing neighborhoods are also open.
\vspace*{0.3cm}


\vspace*{0.5cm}

\noindent
* Corresponding author

\newpage

\thispagestyle{empty}
\pagenumbering{arabic}
\onehalfspacing


\section{Introduction}

Vehicle Routing Problems (VRP) with profits seek to select a subset of customers, each one being associated with a profit, and design up to $m$ vehicle itineraries, starting and ending at a central depot, to visit them. These problems have been the focus of extensive research, as illustrated by the surveys of \cite{Feillet2005}, \cite{Vansteenwegen2010} and \cite{Archetti2013}, mostly because of their difficulty and their numerous practical applications in production planning and logistics \citep{Hemmelmayr2008,Duhamel2009,Tricoire2010,Aras2011,Aksen2012}, manufacturing \citep{Lopez1998,Tang2006a}, robotics, humanitarian relief \citep{Campbell2008} and military reconnaissance \citep{Mufalli2012}, among others.

Three main settings are usually considered in the literature \citep{Chao1996,Archetti2008d,Chu2005,Bolduc2008}:  
profit maximization under distance constraints, called Team Orienteering Problem (TOP);
maximization of profit minus travel costs under capacity constraints, called Capacitated Profitable Tour Problem  (CPTP); and the so-called VRP with Private Fleet and Common Carrier (VRPPFCC), in which customers can be delegated to an external logistics provider, subject to a cost.

To address these problems, we propose to conduct the search on an \emph{exhaustive} solution representation which only specifies the assignment and sequencing of all customers to vehicles\blue{, without deciding which customers are selected in practice. For each route examined during the search}, a \textsc{Select} algorithm, based on a Resource Constrained Shortest Path (RCSP), performs the optimal selection of customers within this sequence and evaluates real route costs.
We then introduce a new Combined Local Search (CLS) working on this solution representation, exploring an exponential set of solutions of VRP with Profits (VRPP) obtained from one standard VRP move with an exponential number of possible combinations of implicit insertions and removals of customers, in pseudo polynomial time.
Pruning techniques are applied to reduce the number of arcs in the shortest-path from $O(n^2)$ to $O(H n)$, where $H \ll n$ is a sparsification parameter. Bi-directional dynamic programming and pre-processing methods are also proposed to solve efficiently the successive RCSPs issued from the local search. This allows to decrease further the amortized complexity of RCSP resolution from $O(B H n)$ down to $O(BH^2)$ for inter-route moves, and $O(B^2H^2)$ for intra-route moves, where $B$ is the average number of labels at each node. As demonstrated by our computational experiments, $B$ remains usually sufficiently small to allow for efficient computations.

The contributions of this work are the following.
1) A new large neighborhood is introduced for vehicle routing problems with profits.
2) Pruning and re-optimization techniques are proposed to perform an efficient search, enabling to decrease the move evaluation complexity by a quadratic factor.
3) These neighborhoods are tested within three heuristic frameworks, a local-improvement procedure, an iterated local search, and a hybrid genetic search.
4) The resulting methods address the three mentioned problems in a unified manner.
5) State-of-the-art results are produced for these settings.
6) Even the simplest local-improvement procedure built on this neighborhood demonstrates outstanding performances on extensively-studied TOP benchmark instances.

\section{Problem statement and unification}

Let $\cG=(\cV, \cE)$ be a complete undirected graph with $|\cV|=n+1$ nodes.
Node $v_0 \in \cV$ represents a depot, where a fleet of $m$ identical vehicles is based.
The other nodes $v_i \in \cV \backslash \{v_0\}$, for $i \in \{1, \ldots , n\}$, represent the customers, characterized by demand $q_i$ and profit $p_i$. Without loss of generality, $q_0 = p_0 = 0$. Edges $(i,j) \in \cE$ represent the possibility of traveling directly from a node $v_i \in \cV$ to a different node $v_j \in \cV$ for a distance/duration $d_{i j}$. In this complete graph, distances are assumed to satisfy the triangle inequality. \blue{Several previous works on TOP have also considered a distinct depot origin and depot destination. This can be modeled by considering $d_{0 i} \neq d_{i 0}$, and thus with an asymmetric~distance matrix.}

The objective of the TOP is to find up to $m$ vehicle routes $\sigma_k = (\sigma_k(1), \dots, \sigma_k(|\sigma_k|))$ for $k \in \{1,\dots,m\}$ starting and ending at the depot, such that the total collected prize $Z_\textsc{top}$ (Equation \ref{TOP1}) is maximized, the sum of traveled distance on any route $\sigma_k$ is smaller than $D$ (Equation \ref{TOP2}), and each customer is serviced at most once.
\begin{align}
Z_\textsc{top} = &\sum\limits_{k=1}^m \sum\limits_{i=1}^{|\sigma_k|-1} p_{\sigma_k(i)} \label{TOP1} \\
 & \sum\limits_{i=1}^{|\sigma_k|-1}  d_{\sigma_k(i)\sigma_k(i+1)} \leq D & k \in \{1,\dots,m\} \label{TOP2}
\end{align}

In the CPTP, \blue{the objective is to produce up to $m$ vehicle routes so as to} maximize the total profit minus travel distance (Equation \ref{CPTP1}), and any route $\sigma_k$ is subject to a capacity constraint (Equation \ref{CPTP2}).

\begin{align}
 Z_{\textsc{cptp}} = & \sum\limits_{k=1}^m \sum\limits_{i=1}^{|\sigma_k|-1} \left\{ p_{\sigma_k(i)} -  d_{\sigma_k(i)\sigma_k(i+1)} \right\} \label{CPTP1} \\
& \sum\limits_{i=1}^{|\sigma_k|-1} q_{\sigma_k(i)} \leq Q &  k \in \{1,\dots,m\} \label{CPTP2}
\end{align}

Finally, in the VRPPFCC, each customer $v_i$ is also associated with an outsourcing cost $o_i$, which is paid if the customer is not serviced. Reversely, this outsourcing cost can be viewed as a profit for customer service, leading to the maximization objective of Equation (\ref{VRPPFCC1}), in which we define the constant $O = \sum_{i=1,\dots,n} o_i$. Each route is also subject to a capacity constraint (Equation \ref{VRPPFCC2}).
\begin{align}
Z_\textsc{vrpfcc} = & \sum\limits_{k=1}^m \sum\limits_{i=1}^{|\sigma_k|-1} \left\{ o_{\sigma_k(i)}  - d_{\sigma_k(i)\sigma_k(i+1)}   \right\} - O  \label{VRPPFCC1} \\
& \sum\limits_{i=1}^{|\sigma_k|-1} q_{\sigma_k(i)} \leq Q &  k \in \{1,\dots,m\} \label{VRPPFCC2}
\end{align}

\begin{proposition}
\label{unification}
TOP, CPTP, and VRPPFCC are all special cases of a two-resources vehicle routing problem with profits (2-VRPP).
In this problem, any arc (i,j) is associated with a resource consumption $r_{ij} \in \Re^+$ and a profit $p_{ij} \in \Re$.
The objective is to build $m$ or less routes, to maximize the total profit (Equation \ref{general1}) while respecting resource constraints on all routes (Equation \ref{general2}).
\begin{align}
Z_{\textsc{2-VRPP}} & =  \sum\limits_{k=1}^m \sum\limits_{i=1}^{|\sigma_k|-1} p_{\sigma_k(i)\sigma_k(i+1)}  \label{general1} \\
& \sum\limits_{i=1}^{|\sigma_k|-1} r_{\sigma_k(i)\sigma_k(i+1)}  \leq R&  k \in \{1,\dots,m\} \label{general2}
\end{align}

\noindent
The reformulation is done as follows for each problem:

\begin{tabular}{l@{\hspace{1cm}}l@{\hspace{1cm}}l@{\hspace{1cm}}l}
\\
TOP: & $r_{ij} = d_{ij}$ & $R = D$ & $p_{ij} = p_i$ \\
CPTP: & $r_{ij} = \frac{q_i}{2}  + \frac{q_j}{2}  $ & $R = Q$ &  $p_{ij} = p_i  - d_{ij}$ \\
VRPPFCC: & $r_{ij} = \frac{q_i}{2}  + \frac{q_j}{2} $ &  $R = Q$ & $p_{ij} = o_i - d_{ij} $ \\
\\
\end{tabular}

\end{proposition}

\begin{proposition}
The resource consumptions $r_{ij}$ satisfy the triangle inequality, i.e. $r_{ij} \leq r_{ik} + r_{kj}$ for any $k \in \{0,\dots,n\} - \{i,j\}$.
\end{proposition}

\begin{proof}
In the TOP, CPTP and VRPPFCC, distances $d_{ij}$ are assumed to satisfy the triangle inequality. In addition, the demand $q_i$ is non-negative for any $i$. Hence, for any $(i,j)$ and $k \in \{0,\dots,n\}-\{i,j\}$,
\begin{itemize}[nosep]
\item for the TOP, $r_{ij} = d_{ij} \leq d_{ik} + d_{kj} \leq r_{ik} + r_{kj}$, and
\item for the CPTP and VRPPFCC, $r_{ij} = \frac{q_i}{2} + \frac{q_j}{2}   \leq  \frac{q_i}{2}  +  q_k +  \frac{q_j}{2}  \leq  r_{ik} + r_{kj}$.
\end{itemize}
\end{proof}

\blue{
The previous properties enable to reduce all three considered problems to the 2-VRPP.
The proposed methodology covers this general case.}

\section{Related literature}

VRPPs have been the subject of a well-developed literature since the 1980s.
The TOP, CPTP and VRPPFCC are NP-hard, and the current exact methods \citep{Butt1999,Boussier2007,Archetti2013a} can solve some instances with up to 200 customers, but mostly when the number of visited customers in the optimal solution remains rather small (less than 50). Heuristics are currently the method of choice for larger problems.

Heuristics and metaheuristics for the TOP have been considerably studied in the past years, and thanks to the rapid availability of common benchmark instances \citep{Chao1996}, a wide range of metaheuristic frameworks have been tested and compared. Neighborhood-centered methods \citep{Vidal2012a} have been generally privileged over population-based search.
\cite{Tang2005} proposed a tabu search with adaptive memory, exploiting both feasible and infeasible solutions in the search process.
\cite{Archetti2006a} introduced a rich family of metaheuristics based on tabu or Variable Neighborhood Search (VNS). The impact of different jump, penalization strategies, and feasibility restoration methods was assessed.
\cite{Ke2008} developed ant-colony optimization (ACO) techniques, and studied four alternative solution-construction approaches: sequential, deterministic-concurrent, random-concurrent, and simultaneous.
\cite{Bouly2009} introduced a hybrid genetic algorithm (GA) based on giant-tour solution representation, which was hybridized later on with particle-swarm optimization (PSO) in \cite{Dang2011}.
\cite{Vansteenwegen2009b} proposed a guided local search, and a path relinking approach was presented in \cite{Souffriau2010}.
Finally, a multi-start simulated annealing method was introduced in \cite{Lin2013}.

The multi-vehicle version of the CPTP has been considered recently in \cite{Archetti2008d}.
The authors extended the tabu and VNS of \cite{Archetti2006a} for this capacity-constrained setting and introduced new benchmark instances.
In addition, the VRPPFCC has been first studied in \cite{Chu2005}, and solved by means of a savings-based constructive procedure and local-improvement.
\cite{Bolduc2007a} and \cite{Bolduc2008} introduced local-search procedures with advanced \textsc{4-opt*}, \textsc{2-opt*}, and \textsc{2-add-drop} movements, as well as perturbation techniques for diversification. \cite{Cote2009a} proposed an efficient tabu search heuristic based on similar moves, which also exploits penalized infeasible solutions during the search. 
Another tabu search, complemented by ejection chains, was described in \cite{Potvin2011}. The latter methods produces results of higher quality, but the ejection chains tend to increase the computational effort.
Finally, \cite{Stenger2012a,Stenger2013} developed an adaptive VNS with cyclic improvements\blue{, and also considered problem extensions with multiple depots and non-linear outsourcing costs.}

Several other VRP with profits have been addressed in the literature, notably with time windows \citep[][among others]{Vansteenwegen2009a,Labadie2010a,Labadie2012a,Lin2012} and in presence of multiple planning periods \citep{Tricoire2010,Zhang2012a}.
Similarly to a wide majority of VRP publications, recent research has been focused for the most part on finding more sophisticated metaheuristic strategies rather than improving the low-level neighborhood structures, which remain the same since many years.
The goal of our paper is to break with this trend by introducing a new family of large neighborhoods. These neighborhoods can be applied in any metaheuristic framework, in possible cooperation with other fast and simple neighborhood structures.

\section{New large neighborhoods for VRP with profits}


All the previously-mentioned efficient metaheuristics rely on local-search improvement procedures to achieve high quality solutions. Most common neighborhoods include separate moves for changing the selection of customers with \textsc{Insert}, \textsc{Remove} or \textsc{Replace} moves, and changing the assignment and sequencing of customer visits with \textsc{Swap}, \textsc{Relocate}, \textsc{2-opt}, \textsc{2-opt*} or \textsc{Cross}. We refer to \cite{Feillet2005} and \cite{Vidal2012a} for a description of these classic neighborhoods.
However, neighborhoods which consider separately the changes of selection and sequencing/assignment may overlook a wide range of simple solution improvements, such as moves on sequences such as \textsc{Swap} or \textsc{Relocate} with a combined \textsc{Insert} of a customer in the same route.

\subsection{Implicit customer selection}
\label{sub-implicit}

We introduce a new neighborhood of exponential size which can be searched in pseudo polynomial time. Two main concepts are exploited: an exhaustive solution representation, and an implicit selection of customers.

VRPPs involve three families of decisions: a \emph{selection} of customers to be visited, the \emph{assignment} of selected customers, and the \emph{sequencing} of customers for each vehicle. In an exhaustive solution, the assignment to vehicles and sequencing \emph{of all customers} are specified, without considering whether they are selected or not. This representation is identical to a complete VRP solution. Some routes may thus exceed the resource constraints, and some may not be profitable, e.g., when off-centered customers with small profits are included.

To retrieve a VRPP solution from an exhaustive solution, a \textsc{Select} algorithm based on a resource constrained shortest path is applied on each route.  \textsc{Select}  retrieves the optimal subsequence of customers, fulfilling the resource constraints, while maximizing the profits.
The goal of this methodology is to keep sequences of non-activated deliveries at promising positions in the solution. These deliveries can become implicitly activated by the \textsc{Select} procedure when a modification, e.g. local-search move such as \textsc{Relocate} or \textsc{swap} is operated \blue{on the exhaustive solution}.

For any route $\sigma$, the selection subproblem is formulated as a resource constrained shortest path in an acyclic directed graph $\cH=(\cV,\cA)$, where $\cV$ contains the $|\sigma|$ nodes visited by a route. Each arc $(i,j) \in \cA$ for $i < j$ is associated with a resource consumption $r_{\sigma(i)\sigma(j)}$ and a profit $p_{\sigma(i)\sigma(j)}$ (c.f. Proposition \ref{unification}).  Subproblems are illustrated in Figure~\ref{selectalgo} for a solution with two routes.

\begin{figure}[h]
\centering
\includegraphics[width=15.5cm]{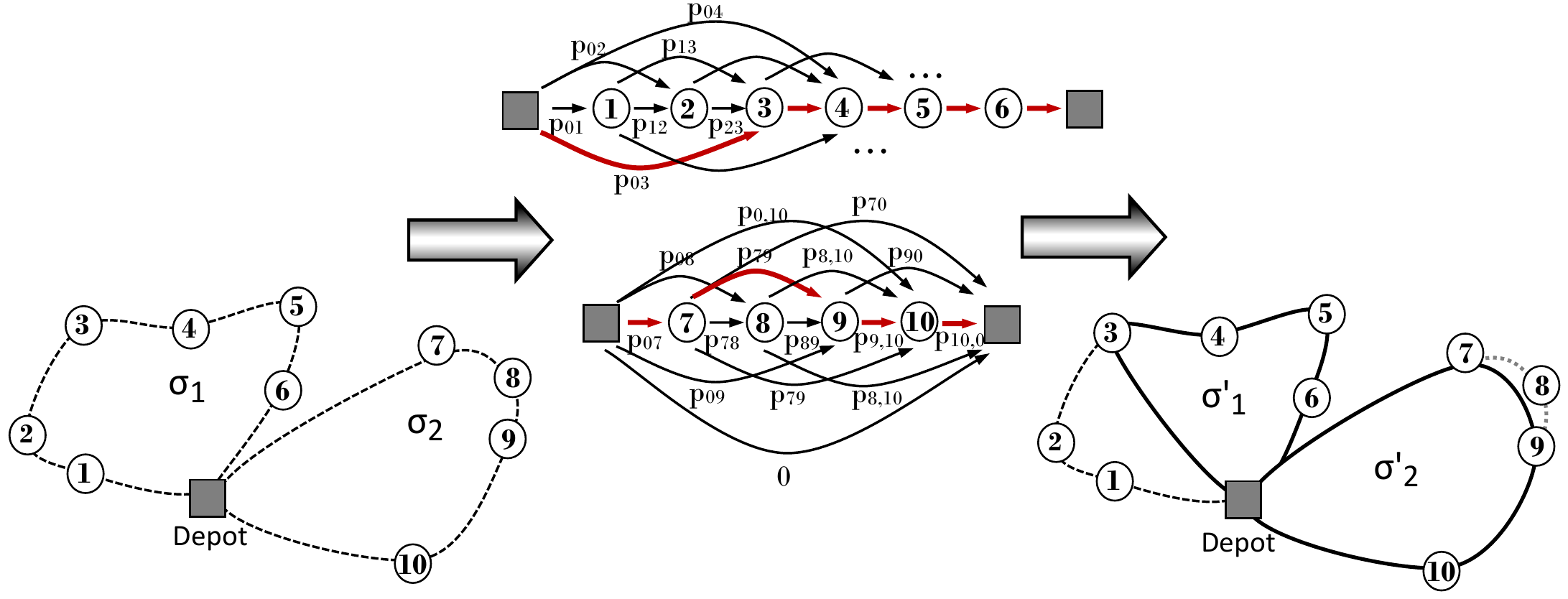}
\caption{From an exhaustive solution to a VRPP solution}
\label{selectalgo}
\end{figure}

Each resource constrained shortest path problem is solved by dynamic programming. A label $s = (s^\textsc{r},s^\textsc{p})$ is defined as a couple (resource,profit). To each node $\sigma(i)$, for $i \in \{1,\dots,|\sigma|\}$ a set of labels $\cS_i$ is associated, starting with $\cS_1 = \{(0,0)\}$ for the node representing the depot. Then, for any $i$, a set of labels $\cS'_{i+1}$ is constructed by considering iteratively any edge $(j,i+1) \in \cA$ and extending all labels of $j$ as in Equation (\ref{extension}).
\begin{equation}
 \label{extension}
\cS'_{i+1} = \bigcup_{j | (j,i+1) \in \cA}  \bigcup_{s_j \in \cS_{j}} (s_j^\textsc{r} + r_{\sigma(j)\sigma(i+1)}, s_j^\textsc{p} + p_{\sigma(j)\sigma(i+1)}) 
\end{equation}

Any infeasible label $s \in \cS'_{i+1}$, such that $s^\textsc{r} + r_{\sigma(i+1)\sigma(0)} > R$ is pruned from $\cS'_{i+1}$. Indeed, resource consumptions on edges satisfy the triangle inequality, and thus the resource consumption for returning to $v_0$ after $\sigma(i+1)$ on any path is greater or equal than $r_{\sigma(i+1)\sigma(0)}$.
All dominated labels of $\cS'_{i+1}$, i.e. labels $(s^\textsc{r},s^\textsc{p}) \in \cS'_{i+1}$ such that there exists $(\bar{s}^\textsc{r},\bar{s}^\textsc{p}) \in \cS'_{i+1}$ with $s^\textsc{r} \geq \bar{s}^\textsc{r}$ and $s^\textsc{p} < \bar{s}^\textsc{p}$ are also removed to yield $\cS_{i+1}$. \\

The resulting \textsc{Select} algorithm is pseudo polynomial, with a complexity of $O(n^2 B)$, where $B$ is an upper bound on the number of labels per node. 


\subsection{Neighborhood search on the exhaustive solution}

\blue{We now consider a combined local search procedure which applies classic VRP moves on the exhaustive solution representation. Evaluating the profitability of any move requires to use the \textsc{Select} algorithm on the newly created routes to find the optimal selection of customers. As such, insertions and removals of customers are implicitly managed within the selection process, instead of being explicitly considered by the local search.}

\blue{
In the proposed heuristics, we consider the standard VRP moves \textsc{2-opt}, \textsc{2-opt*}, \textsc{Relocate}, \textsc{Swap} and \textsc{Cross} of a maximum of two customers \citep{Vidal2012a}. These moves are tested between pairs of nodes $(v_i,v_j)$ such that $v_j$ is among the $\Gamma$ closest nodes of $v_i$ (similarly to \citealt{Johnson1997} and \citealt{Toth2003}). Moves 
are enumerated in a random order, any improvement being directly applied. The method stops whenever all moves have been tried without success.}

\subsection{Illustrative example}

\begin{figure}[htb]
\begin{minipage}[c]{.2\linewidth}
\begin{scriptsize}
\scalebox{0.9}
{
\begin{tabular}{|l@{\hspace{0.1cm}}c@{\hspace{0.2cm}}c@{\hspace{0.2cm}}c|}
\hline 
i&$r_{0,i}$&$r_{i-1,i}$&$p_{ij} \ \forall j$\\
\hline 
1&15&--&10\\
2&25&30&15\\
3&15&20&15\\
4&15&20&10\\
5&20&25&12\\
6&15&10&15\\
7&20&15&15\\
8&25&15&12\\
9&25&20&15\\
10&15&35&15\\
\hline 
\multicolumn{4}{|c|}{$R = 100$} \\
\multicolumn{4}{|c|}{$r_{7,9} = 25$, $r_{10,0} = 15$ } \\
\multicolumn{4}{|c|}{all other consumptions = $+\infty$} \\
\hline 
\multicolumn{4}{c}{\strut} \\
\hline 
$\sigma$&\multicolumn{2}{c}{$R(\sigma)$}&$P(\sigma)$\\
\hline 
(3,4,5,6)&\multicolumn{2}{c}{85}&52\\
(7,9,10)&\multicolumn{2}{c}{95}&45\\
(1,2,3,4)&\multicolumn{2}{c}{100}&50\\
(6,7,8,9)&\multicolumn{2}{c}{90}&57\\
\hline 
\end{tabular} 
}
\end{scriptsize}
\end{minipage} \hfill
\begin{minipage}[c]{.76\linewidth}
\centering
\includegraphics[width=11cm]{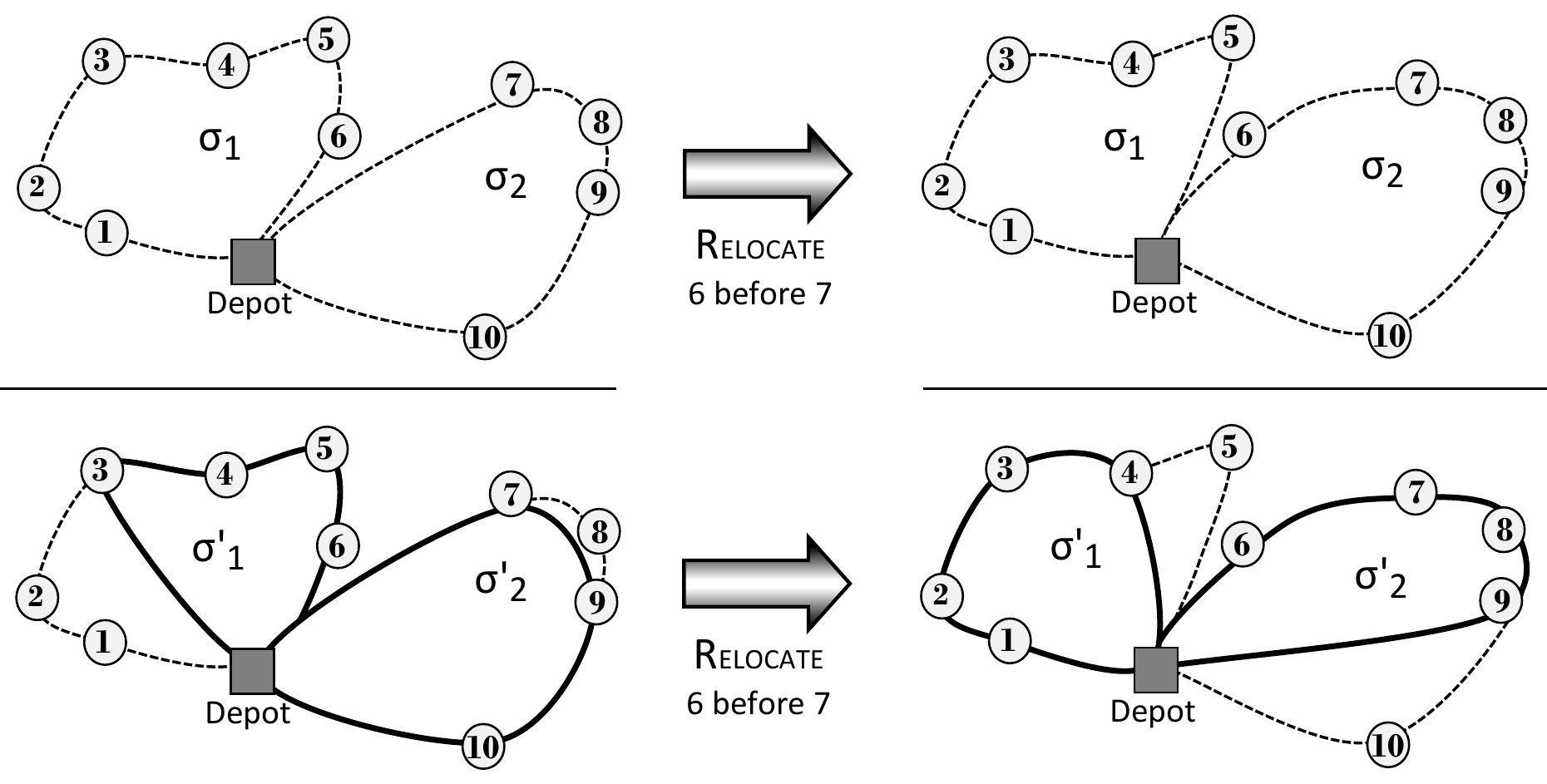}
\end{minipage}
\caption{\textsc{Relocate} move on the exhaustive solution representation, and impact on the associated VRPP solution}
\label{CompoundMove}
\end{figure}


Figure \ref{CompoundMove} illustrates a simple \textsc{Relocate} move of customer $v_6$ before $v_7$ on the exhaustive solution (top of the figure), and its impact on the associated VRPP solution (bottom of the figure).
The initial VRPP solution, on the left, was a local-optimum of all classic neighborhoods. The \textsc{Relocate} move has a dramatic impact on the associated VRPP solution, since the \textsc{Select} algorithm operates different choices as a consequence, here a compound \textsc{Removal} of $v_5$ and $v_{10}$ and \textsc{Insert} of $v_1$ and $v_2$ before $v_3$. As a result, a feasible solution with higher profit is attained.

\subsection{Connections with other large neighborhoods}
\label{aux-obj}

\blue{The proposed methodology is distinct from the \emph{Split} algorithm \citep{Beasley1983,Prins2004}, which aims to find the best segmentation of a permutation of customers into separate routes by inserting visits to the depot.}
In particular, we applied these neighborhoods within a Unified Hybrid Genetic Search with giant-tour ``exhaustive'' solution representation (Section \ref{exp-description-meta}). In this context, the Split algorithm can be implemented as a by-product of route evaluation procedures \citep{Vidal2012b}, and it assumes the twofold task of inserting depot visits and selecting customers.

These concepts are also clearly distinct from efficient ejection chains \citep{Glover1996} and cyclic improvement procedures based on the search for a positive-cost cycle in an auxiliary \emph{improvement graph}. In our case, multiple compound \textsc{Insert} or \textsc{Remove} can implicitly arise from a move. The two methodologies could even be combined together by considering ejection chains on the exhaustive representation, and thus, searching a positive cycle in an auxiliary graph in which each arc cost has been obtained by means of the \textsc{Select} algorithm. This is left as a perspective of research.

\subsection{Hierarchical objective}
\label{aux-obj}

The total profit is usually not straightforward to improve with local and even large-neighborhood search, since it requires the ability to \emph{create space} in order to save resources and deliver more profitable customers. An improvement of the objective is thus often the result of several combined moves.
For the TOP in particular, the search space has a \emph{staircase} aspect since multiple solutions have equal profit. This is usually not well-suited for an efficient search.
To avoid this drawback and drive non-activated customers towards promising locations, the total distance of the current exhaustive solution is considered as a secondary objective (Equation \ref{second}) with very small weight $\omega \ll 1$. 

\begin{equation}
\label{second}
Z' = Z_{\textsc{2-vrpp}}(\sigma) + \omega \sum_{\sigma \in \cR}  \sum\limits_{i \in \{1, \dots, |\sigma|-1\}} d_{\sigma(i)\sigma(i+1)}
\end{equation}

As a result, even if no improving move for the primary objective of the VRPP can be found from an incumbent solution, the neighborhood search will re-arrange the deliveries in better positions. This may open the way to new improvements of the main objective at a later stage, without requiring any solution deterioration.

\subsection{Speed-up techniques}
\label{speed-ups}

Solving from scratch each such resource constrained shortest path leads to computationally expensive move evaluations in $O(n^2 B)$. Such near-quadratic complexity may not be acceptable in recent neighborhood-based heuristic searches which rely on a considerable number of route evaluations.

\subsubsection{Graph sparsification} 
\label{sparsification}

To reduce this complexity, we propose to prune several arcs in the shortest-path graph, and rely on pre-processing and bi-directional dynamic programming. For a given \emph{sparsification parameter} $H \in \mathbb{N^+}$, only arcs $(i,j)$, with $(i < j)$ satisfying Equation (\ref{pruning})~are~kept.
\begin{equation}
\label{pruning}
j < i + H \text{ or } i=0 \text{ or } j=|\sigma|
\end{equation}

\blue{$H$ imposes an upper bound on the number of consecutive non-activated deliveries which can arise in a VRPP solution.}
The number of non-activated deliveries located just after or before a depot still remains unlimited, thus guaranteeing the existence of a feasible solution. This sparsification parameter is illustrated on Figure \ref{speedup1}.

\begin{figure}[h]
\centering
\includegraphics[width=12cm]{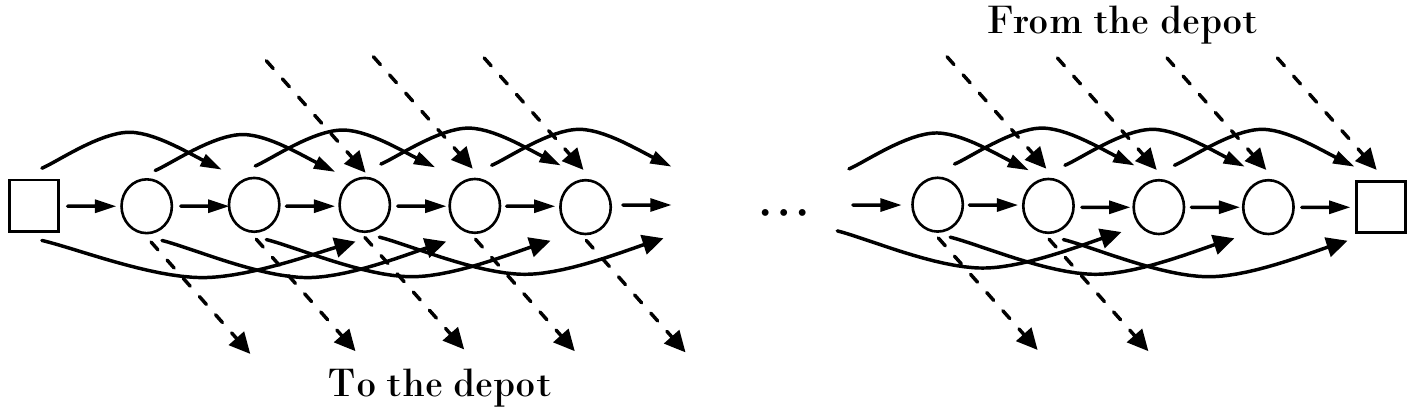}
\caption{Shortest path graph after sparsification (H=3)}
\label{speedup1}
\end{figure}

\blue{
\begin{proposition}
After sparsification, the number of arcs $|\cA'|$ in the new auxiliary graph $\cH' = (\cV,\cA')$ becomes $O(n+n+Hn) = O(nH)$. The first two terms are the arcs originating and ending at the depot, and the term $Hn$ relates to the limited number of intermediate arcs. The complexity of \textsc{Select} becomes $O(n H B)$.
\end{proposition}
}


\subsubsection{Labels pre-processing and route evaluations by concatenation} 

Solutions resulting from classic VRP moves on an incumbent exhaustive solution can all be assimilated to recombinations of a bounded number of subsequences of consecutive visits from this solution. This property is thoroughly discussed in \cite{Vidal2011b,Vidal2012b}. As a consequence, we propose advanced move evaluation methods which pre-process additional information on subsequences of consecutive customers from the incumbent solution to reduce the complexity of move evaluations.

There are $O(n^2)$ subsequences of consecutive customers in the incumbent solution. Pre-processing is done on these sequences, and the values are updated through the search whenever some routes of the incumbent solution are modified. As will be shown in the following, the effort related to information preprocessing is compensated by the reduction in computational effort related to move evaluations.
In this section, we explain how this \blue{concept} can be exploited to reduce even further the complexity of the \textsc{Select} algorithm for VRPPs.

\blue{For ease of presentation, we will call \emph{``starting node of $\sigma$''} any node among the $H$ first nodes of $\sigma$, and \emph{``finish node of $\sigma$''} any node among the H last nodes of $\sigma$.}
As in previous papers, we explain 1) the nature of the information, 2) how to do the preprocessing and 3) how to create an advanced \textsc{Select} algorithm which evaluates moves as a concatenation of known subsequences, using the information of each subsequence. \\

\textbf{Information collected and preprocessing.}  For any subsequence $\sigma$ of \blue{consecutive} nodes from the incumbent solution, the following information is pre-processed:
\begin{itemize}
\item Set of non-dominated labels $S_{ij}(\sigma)$ \blue{obtained from a resource constrained shortest path from a starting node $i$ to a finish node $j$ of $\sigma$. This information is computed for $i \in \{1,\dots,\min(H,|\sigma|)\}$ and $j \in \{\max(|\sigma|-H+1,1),\dots,|\sigma|\}$.}
\item Set of non-dominated labels $S^{\textsc{end}}_{i}(\sigma)$  \blue{obtained from a resource constrained shortest path in $\sigma$ from a starting node to the ending depot, for $i \in \{1,\dots,\min(H,|\sigma|)\}$.}
\item  Set of non-dominated labels $S^{\textsc{beg}}_{j}(\sigma)$  \blue{obtained from a resource constrained shortest path in $\sigma$ from the starting depot to a finish node $j$, for $j \in \{\max(|\sigma|-H+1,1),\dots,|\sigma|\}$.}
\item  \blue{Maximum} profit $P(\sigma)$ of a feasible path within $\sigma$, which starts from the depot, visits a subset of customers in $\sigma$, and comes back to the depot.
\end{itemize}

 \blue{Preprocessing these values for all sequences $\sigma$ requires $O(n^2 H B)$ elementary operations. Indeed, the non-dominated labels associated to the all-pairs resource constrained shortest path in graph $\cH'$, between any node $i$ and $j$ (including the depot nodes), provide the required information. These labels are computed by forward labeling in $O(n \cdot n H B)$.}

\begin{figure}[h]
\centering
\includegraphics[width=14.5cm]{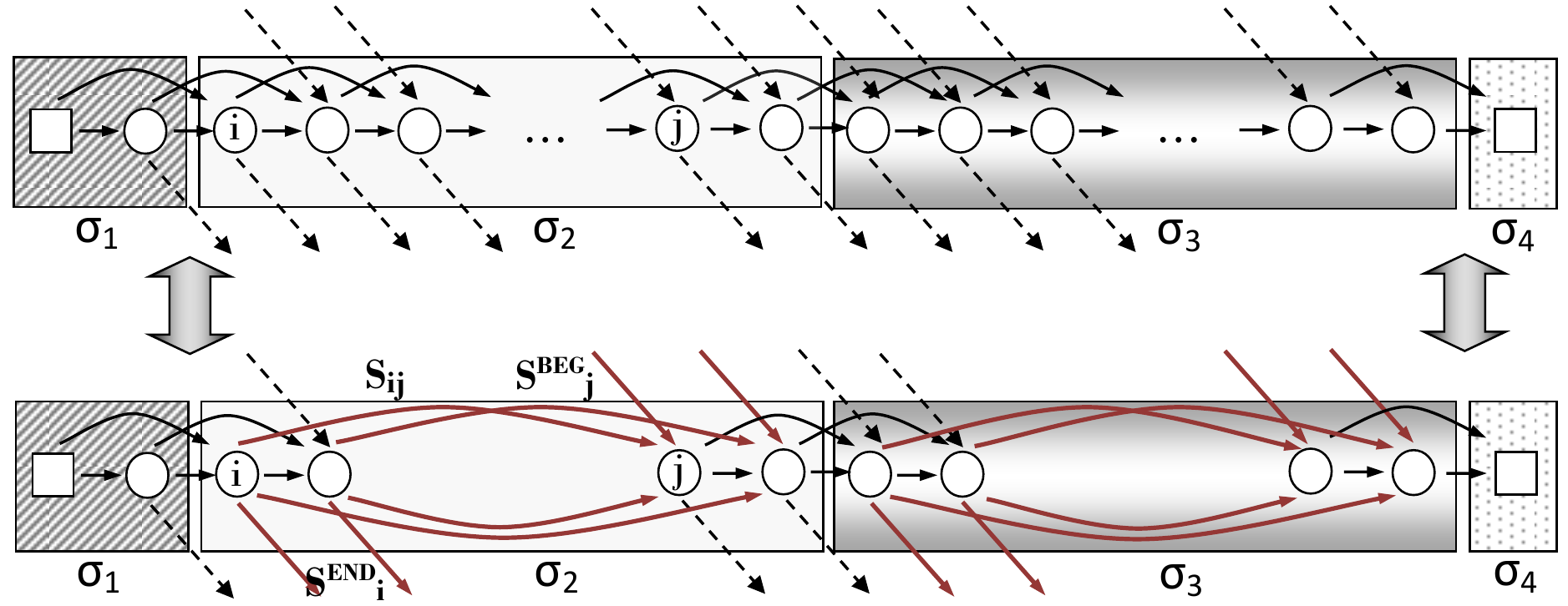}
\caption{Using pre-processed information to reduce the graph (H=2)}
\label{speedup2}
\end{figure}

\blue{Now, the contribution of this preprocessed information is illustrated in Figure~\ref{speedup2}.
The top of the figure illustrates a shortest path problem in graph $\cH'$, for a new route involving four known subsequences. The edges in dotted lines are the ones originating or ending at the depot nodes. Here $H=2$, and thus it is possible to skip only one node at a time.}

\blue{Remark that only arcs originating from the depot or finish nodes are connected to starting nodes in the next subsequence. Thus, any optimal path in the graph $\cH'$ visiting customers from more than two distinct sequences $\{\sigma_x,\dots,\sigma_y\}$ must visit at least one starting node or finish node in each of these subsequences $\sigma_k \in \{\sigma_x,\dots,\sigma_y\}$. As such, the partial path within a subsequence $\sigma_k$, either linking a starting node $i$ to a finish node $j$, or a starting node $i$ to the depot, or the depot to a finish node $j$, is represented by a label in $S_{ij}(\sigma_k)$, $S^{\textsc{end}}_{i}(\sigma_k)$ or $S^{\textsc{beg}}_{j}(\sigma_k)$, respectively. These pre-processed labels can be introduced as new arcs in the graph, and they replace previous intermediate arcs and nodes within the sequence. This leads to a new shortest path problem in a graph $\cG''$ illustrated in the bottom of the figure.}

The new \emph{reduced} graph $\cG'' = (\cV'', \cA'')$ is such that $|\cA''| = O(M H^2)$ arcs and $|\cV''| =  O(M H)$ nodes. $M$ is a bounded value which stands for the number of subsequences. Classic intra-route moves for the VRP are such that $M \leq 3$, and all inter-route moves are such that $M \leq 5$ \citep{Vidal2011b,Vidal2012b}. The quantity of arcs $O(H^2)$ does not depend anymore on the number of customers $n$. Still, the newly-created arcs, for any sequence $\sigma$, are not associated with a single profit and resource consumption, but rather to a set of non-dominated (resource, profit) couples $(s^\textsc{r},s^\textsc{p}) \in S_{ij}(\sigma)$, such that $|S_{ij}(\sigma)| = O(B)$. 

\begin{proposition}[\textbf{Concatenation -- general}]
\label{prop-concatN}
\label{concat-general}
The optimal profit $P(\sigma_1 \oplus \dots \oplus  \sigma_M)$ of \textsc{Select}, for a route \blue{made of} $M$ concatenated sequences $\sigma_1 \oplus \dots \oplus  \sigma_M$, is the maximum of the profit $\bar{P}(\sigma_1 \oplus \dots \oplus  \sigma_M)$ of the resource constrained shortest path in $\cG''$, and the maximum profit $P(\sigma_i) $ of a feasible path within $\sigma_i$ exclusively, for $i \in \{1,\dots,M\}$:
\begin{equation}
P(\sigma_1 \oplus \dots \oplus  \sigma_M) = \max \{ \bar{P}(\sigma_1 \oplus \dots \oplus  \sigma_M),  \max\limits_{i \in \{1,\dots,M\}} P(\sigma_i) \}
\end{equation}
$\bar{P}(\sigma_1 \oplus \dots \oplus  \sigma_M)$ can be evaluated in $\Phi_{\textsc{C-M}} =  O(M H^2 B^2)$.
\end{proposition}

Indeed, solving the shortest path on the reduced graph enables to find the best solution containing at least one \blue{starting or finish} node. A better solution can exist, for each sequence $\sigma_i$, by visiting exclusively some inside nodes which have been eliminated in the reduced graph. The cost of these specific solutions is included in the preprocessed value $P(\sigma_i)$.
Now, the reduced graph $\cG''$ has $O(M H^2)$ arcs, and each arc is associated with up to $B$ non-dominated (resource, profit) couples. Solving the resource constrained shortest path on this graph to find $\bar{P}(\sigma_1 \oplus \dots \oplus  \sigma_M)$ is equivalent to solving a resource constrained shortest path on a multi-graph, where arcs are duplicated $B$ times to take into account the different (resource,profit) combinations. This can be done as previously by means of a forward labeling method in $O(M H^2 B^2)$. \\

\blue{This algorithm derived from Proposition \ref{concat-general} is general in the sense that it can be used to evaluate a route made of any number of concatenated subsequences. Yet,} it is used in our method only to evaluate moves which involve more than three subsequences, such as intra-route \textsc{2-opt},  \textsc{Swap} and \textsc{Relocate} moves between two positions in the same route. The other -- more numerous -- inter-route \textsc{Swap}, \textsc{Relocate}, \textsc{Cross} and \textsc{2-opt*}  moves involve less than three subsequences \citep{Vidal2012a} as illustrated in Figure~\ref{speedup4}.
Proposition \ref{prop-concat3} provides a way to evaluate these moves even more efficiently by means of \blue{bi-directional} dynamic programming.

\begin{figure}[h]
\centering
\includegraphics[width=14cm]{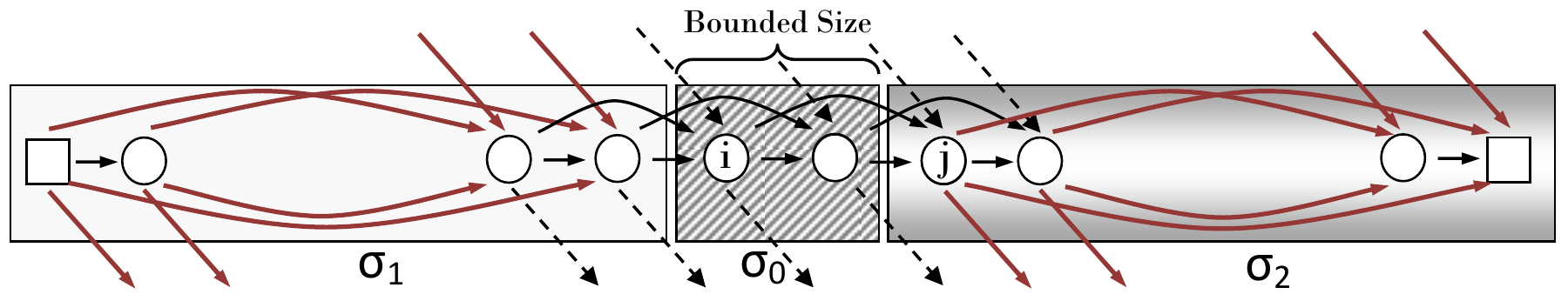}
\caption{Route resulting from a concatenation of three subsequences}
\label{speedup4}
\end{figure}

\begin{proposition}[\textbf{Concatenation -- 3 subsequences}]
\label{prop-concat3}
The optimal cost $P(\sigma_1 \oplus \sigma_0 \oplus \sigma_2)$ of \textsc{Select}, for a route \blue{made of} three concatenated subsequences $\sigma_1 \oplus \sigma_0 \oplus \sigma_2$ such that $\sigma_0$ contains a bounded number of customers can be evaluated in
$\Phi_{\textsc{C-3}} =  O(H^2 B).$
The same complexity is achieved for a concatenation of two sequences $\sigma_1$ and $\sigma_2$.
\end{proposition}

Indeed, the maximum of $ \{P(\sigma_1), P(\sigma_2) \}$ gives in $O(1)$ the \blue{best cost of a path that does not visit any starting or finish node in $\sigma_1$ or $\sigma_2$. Then, the best cost of a path visiting at least one starting or finish node can be evaluated as follows.}

\begin{enumerate}
\item We recall that the set of non-dominated labels of a resource constrained shortest path going from the starting depot to a \blue{finish node $k$ of the first sequence $\sigma_1$} is given by $S^{\textsc{beg}}_{k}(\sigma_1)$.

\item The algorithm propagates these labels iteratively to the next $|\sigma_0|$ nodes, obtaining the set of labels $\bar{S}_i(\sigma_{1,0})$ associated with the shortest path between the depot and the finish nodes of $\sigma_{1,0} = \sigma_1 \oplus \sigma_0$, for $i \in \mathcal{I} = \{ \max(|\sigma_1| + |\sigma_0| - H + 1,1), \dots, |\sigma_1| + |\sigma_0|\}$. All arcs in $|\sigma_0|$ are simple arcs, such that each propagation is done in $O(HB)$. The propagation to the finish depot for each $i$ is also evaluated, leading in addition to the optimal cost $\bar{P}(\sigma_{1,0})$ of a path servicing at least one node in both $\sigma_1$ and $\sigma_0$ and returning to the end depot without visiting $\sigma_2$.

\item Finally, the best cost $\bar{P}(\sigma_{1,0}  \oplus \sigma_2)$  of a path servicing nodes in both $\sigma_{1,0}$ and $\sigma_2$ is equivalent to finding, for each possible arc $(i,j) \in \cA''$ between $\sigma_{1,0}$ and $\sigma_2$, the best pair of labels $(s_k^\textsc{r},s_k^\textsc{p}) \in \bar{S}_i(\sigma_{1,0})$ and $(s_l^\textsc{r},s_l^\textsc{p}) \in S^{\textsc{end}}_{j}(\sigma_2)$ which maximizes the total profit while respecting resource constraints as expressed in Equation~(\ref{junction}).
\begin{equation} 
\label{junction}
\bar{P}(\sigma_{1,0}  \oplus \sigma_2) = 
\max_{i \in \mathcal{I} }
\max_{j \in \{1,\dots,\min(H,|\sigma_2|)\}}
 \left\{
\begin{aligned}
 \max_{k, l} \ &  \ s_k^\textsc{p} + p_{\sigma_{1,0}(i)\sigma_2(j)} + s_l^\textsc{p} \\
\text{s.t.}  \ & \ s_k^\textsc{r} + r_{\sigma_{1,0}(i)\sigma_2(j)} + s_l^\textsc{r}  \leq R \\
& (s_k^\textsc{r},s_k^\textsc{p}) \in \bar{S}_i(\sigma_{1,0}) \\
& (s_l^\textsc{r},s_l^\textsc{p}) \in S^{\textsc{end}}_{j}(\sigma_2)
\end{aligned} \right\}
\end{equation}

For any $(i,j)$, the maximum cost \blue{of} Equation (\ref{junction}) can be computed in $O(B)$ by sweeping the labels of $\bar{S}_i(\sigma_{1,0})$ by increasing resource consumption and, in the meantime, the labels of $S^{\textsc{end}}_{j}(\sigma_2)$ by decreasing resource consumption.

The computational complexity of the method is $O(|\sigma_0| HB) = O(HB)$ for Phase 2 and $O(H^2 B)$ for Phase 3, leading to an overall complexity of $O(H^2 B)$.
\end{enumerate}

The final cost returned by the algorithm is 
\begin{equation}
P(\sigma_1 \oplus \sigma_0  \oplus \sigma_2) = \max \{ \bar{P}(\sigma_{1,0}  \oplus \sigma_2), \bar{P}(\sigma_{1,0}), P(\sigma_1), P(\sigma_2) \}.
\end{equation}

For the case involving only two sequences $\sigma_1$ and $\sigma_2$, the same algorithm without Phase~2 yields the same complexity. \\

\blue{Classic inter-route VRP neighborhoods contain $\Theta(n^2)$ moves, for a total evaluation complexity of $\Theta(n^2 H^2 B)$ (c.f. Proposition \ref{prop-concat3}).} The complexity of the preprocessing phase is $O(n^2 H B)$. Hence, move evaluation is dominating in terms of CPU time. This statement is also supported by our empirical analyses, which show that the 
additional preprocessing effort did not have a significant impact on CPU time. 

Finally, intra-route moves are less numerous, in practice there are an average of $\Theta(\frac{n^2}{m})$ such moves to consider where $m$ is number of routes.
The neighborhood evaluation complexity (c.f. Proposition \ref{prop-concatN}) becomes $\Theta(\frac{n^2 }{m} H^2 B^2) = \Theta(\frac{B}{m} n^2  H^2 B)$.
The ratio $\frac{B}{m}$ is usually small and the observed effort spent on intra- and inter-route neighborhoods is of similar magnitude in our experiments.

%
%

\section{Experimental Analyses}

The contribution of this new large neighborhood is assessed by extensive computational experiments with three families of heuristics. The impact of key parameters is investigated on the TOP,  and the performance of each method is evaluated on the TOP, as well as on the CPTP and VRPPFCC relatively to the current state-of-the-art \blue{solution methods}.

\subsection{Heuristics and metaheuristic frameworks}
\label{exp-description-meta}

We selected three particular methods for our experiments: a multi-start local-improvement procedure (MS-LS), which represents one of the most simple scheme in heuristic search but is still at the core of most metaheuristics; a classic neighborhood-centered method such as the ILS of  \cite{Prins2009c}, and finally a more advanced population-based method with diversity management such as UHGS \citep{Vidal2012b}. This way, the contribution of the proposed neighborhoods can be investigated on three notable VRP heuristic frameworks.

\begin{itemize}
\item MS-LS is a straightforward application of the proposed neighborhood search, based on the classic neighborhoods \textsc{2-opt}, \textsc{2-opt*}, \textsc{Relocate}, \textsc{Swap} and \textsc{Cross} exchanges of up to two customers. As described previously, moves are applied between close customers on the exhaustive solution representation. A random order is used for move evaluation, any improving move being directly applied until no improvement can be found in the whole neighborhood. At the end, the resulting solution is a local optimum of the proposed neighborhood. The local-improvement procedure is repeated $\mu$ times from randomly generated initial solutions. The best local optimum constitutes the MS-LS solution.

\item MS-ILS is a direct adaptation of the method of \cite{Prins2009c}. Starting from a random initial solution, $n_C$ child solutions are iteratively generated by applying a shaking operator and the same local search as MS-LS. The best child solution is taken as new incumbent solution for the next iteration. The method is started $n_P$ times. Each run is completed once $n_I$ consecutive iterations have been performed without improvement of the best solution, or when an overall maximum time $T_{max}$ is attained. The overall best solution is finally returned. 

\item Finally, our implementation of UHGS derives directly from the general framework of \cite{Vidal2012b}. The 2-VRPP has been addressed only by adding a new route-evaluation operator, and all other procedures, selection, crossover, education, \emph{Split} algorithm and population-diversity management procedures remain the same. As previously, the population is managed to contain between $\mu^{\textsc{MIN}}$ and $\mu^{\textsc{MIN}}\nobreak+\nobreak\mu^{\textsc{GEN}}$ solutions, and the method terminates whenever $It_{max}$ iterations -- individual generations -- without improvement have been performed or when a CPU time $T_{max}$ is attained. No infeasible solutions need to be used in this context since the evaluation procedure allows for any route size without infeasibility.
\end{itemize}

\subsection{Benchmark Instances}
\label{exp-benchmark}

The performance of these three heuristics built on our new compound neighborhoods is assessed by means of extensive experiments on classic benchmark instances for the considered problems.
The classic TOP instances \citep{Chao1996} are classified into seven sets which include respectively 32, 21, 33, 100, 66, 64 and 102 customers. Each instance set is declined into individual instances with between 2 to 4 vehicles and different duration limits. We consider Sets 4-7, since optimal solutions are systematically obtained on the smaller problems. Other particular instances for which all methods from the recent literature find the optimal solutions have been also excluded from the experiments, only keeping the 157 most difficult ones as in \cite{Souffriau2010}.

The CPTP instance sets of \cite{Archetti2008d} have been derived from the classic VRP instances of \cite{Christofides1979} using different values of capacity $Q$ and fleet size $m$.  We refer to each instance as ``pXX-$m$-$Q$'', where XX is the index of the associated VRP instance.
Finally, we rely for the VRPPFCC on the two sets of instances from \cite{Bolduc2008}: \blue{Set ``CE'' is derived from the instances of \cite{Christofides1979} and includes up to 199 customers, while Set ``G'' includes larger instances with up to 483 customers, derived from \cite{Golden1998}}. Our three methods are compared to the best current metaheuristics in the literature, listed in Table~\ref{BestMethod}.

\begin{table}[htb]
\caption{Methods, references and acronyms}
\label{BestMethod}
\hspace*{-1.2cm}
\scalebox{0.80}
{
\begin{tabular}{|r|rll|}
\hline
\textbf{Problem}&\textbf{Acronym}&\textbf{Method}&\textbf{Authors} \\
\hline
\textbf{TOP}&TMH&Tabu Search&\cite{Tang2005} \\
&GTP&Tabu Search with Penalization Strategy&\cite{Archetti2006a} \\
&GTF&Tabu Search with Feasible Strategy&\cite{Archetti2006a}\\
&FVF&Fast Variable Neighborhood Search&\cite{Archetti2006a}\\
&SVF&Slow Variable Neighborhood Search&\cite{Archetti2006a}\\
&GLS&Guided Local Search&\cite{Vansteenwegen2009c}\\
&ASe&Ant Colony Optimization -- Sequential&\cite{Ke2008}\\
&ADC&Ant Colony Optimization -- Deterministic Concurrent&\cite{Ke2008}\\
&ARC&Ant Colony Optimization -- Random Concurrent&\cite{Ke2008}\\
&ASi&And Colony Optimization -- Simultaneous&\cite{Ke2008}\\
&SVNS&Skewed Variable Neighborhood Search&\cite{Vansteenwegen2009b}\\
&FPR&Fast Path Relinking&\cite{Souffriau2010}\\
&SPR&Slow Path Relinking&\cite{Souffriau2010}\\
&SA&Simulated Annealing&\cite{Lin2013}\\
&MSA&Multi-Start Simulated Annealing&\cite{Lin2013}\\
\hline
\textbf{CPTP}&GTP&Tabu Search with Admissible Strategy&\cite{Archetti2008d}\\
&GTF&Tabu Search with Feasible Strategy&\cite{Archetti2008d}\\
&VNS&Variable Neighborhood Search&\cite{Archetti2008d}\\
\hline
\textbf{VRPPFCC}&SRI&Selection, Routing and Improvement&\cite{Bolduc2008}\\
&RIP&Randomized Construction-Improvement-Perturbation&\cite{Bolduc2008}\\
&TS&Tabu Search&\cite{Cote2009a}\\
&TS2&Tabu Search with Ejection Chains&\cite{Potvin2011}\\
&TS+&Tabu Search with Ejection Chains&\cite{Potvin2011}\\
&AVNS& Adaptive Variable Neighborhood Search &\cite{Stenger2012a}\\
\hline
\end{tabular}
}
\end{table}

\subsection{Parameter setting and sensitivity analyses}

\textbf{General parameter setting.}
The proposed methodology relies on two parameters: the weight $\omega$ of the secondary objective and the sparsification coefficient $H$. The only purpose of $\omega$ is to establish a hierarchy between the real problem objective and the auxiliary route length objective (Section \ref{aux-obj}). From our experiments, setting $\omega = 10^{-4}$ establishes the desired hierarchy without involving numerical precision issues. The impact of the sparsification coefficient $H$ is discussed later in this section.

For the other parameters, we aimed to implement the methods without significant changes from the original papers \citep{Prins2009c,Vidal2012b}. Still, to compare with other authors with similar computational effort, we had to scale the termination criteria and population size, leading to the setting $\mu = 5$ for MS-LS, $(n_{\textsc{P}},n_{\textsc{I}},n_{\textsc{C}}) = (3,10,3)$ for MS-ILS, and $(\mu^{\textsc{MIN}},\mu^{\textsc{GEN}},It_{max}) = (5,10,500)$ for UHGS. The overall CPU time is also limited to $T_{max} = 300s$.\\


\noindent
\textbf{Sensitivity analysis on the sparsification parameter.}
Parameter $H$ is an important element of the newly proposed neighborhoods.  Larger values of this parameter lead to larger neighborhoods but higher CPU time. A good balance is thus desired. To calibrate $H$ and analyze its impact on the method, we ran several tests on the UHGS with values of $H \in \{1,2,3,5,7,10, \infty\}$. Each parameter configuration was tested on the 157 TOP benchmark instances and 10 runs have been conducted. The solution quality  is measured for each instance as a percentage of deviation from the Best Known Solution (BKS) in the literature, $100(z_\textsc{bks}-z)/z_\textsc{bks}$, where $z$ is the profit obtained by the method and $z_\textsc{bks}$ is the profit of the BKS. The best known solutions, including those found in this paper, are included in Table \ref{results-TOP1}.

The results of this calibration, averaged on all instances, are displayed in Table \ref{calib-TOP}. We report for each configuration the average deviation to BKS on 10 runs, the deviation of the best solution of these runs, the number of BKS found out of 157, the average CPU time, the average CPU time to reach the best solution of the run, and the number of labels per customer node in the shortest path subproblems.  Finally, the last line reports the p-values of a Wilcoxon signed-rank test for paired samples, between the average profit obtained with the reference parameter setting $H=3$ and any other setting $x$. The p-value illustrates how likely is the null hypothesis ``that both methods for $H=3$ and $H = x$ perform equally''.

\begin{table}[htb]
\caption{Calibration of the sparsification parameter $H$}
\label{calib-TOP}
\hspace*{-0.3cm}
\scalebox{0.88}
{
\begin{tabular}{|r|p{1.5cm}p{1.5cm}p{1.5cm}p{1.5cm}p{1.5cm}p{1.5cm}p{1.5cm}p{1.5cm}|}
\hline
&H=1&H=2&H=3&H=4&H=5&H=7&H=10&H=$\infty$\\
\hline
Avg Gap&0.090\%&0.022\%&0.020\%&0.035\%&0.036\%&0.055\%&0.087\%&0.177\%\\
Best Gap&0.008\%&0.002\%&0.001\%&0.002\%&0.005\%&0.005\%&0.010\%&0.016\%\\
Nb BKS&144&155&155&153&152&152&148&143\\
Avg Time(s)&136.83&174.32&192.00&211.23&223.54&236.48&247.87&282.81\\
Avg T-Best(s)&84.89&92.54&91.09&103.10&111.19&118.82&130.40&156.86\\
Avg Labels&8.39&29.19&34.70&36.63&36.90&36.50&35.27&32.81\\
&\multicolumn{8}{c|}{\strut}\\
P-value & 2.4E-12 & 0.23 & -- & 1.5E-4 & 2.1E-5 & 5.5E-7  & 2.2E-8 &  3.0 E-11 \\
\hline
\end{tabular}
}
\end{table}

From these tests, it appears that the configuration $H=3$ produces solutions of significantly better quality than any configuration with $H \in \{1, 4, 5, 7, 10, +\infty\}$. This configuration attains 155/157 best known solutions. It requires $40.3\%$ more CPU time than with $H=1$, but only $7.3\%$ more CPU time to reach the best solution. The average number of non-dominated labels kept per node is greater by a factor $4.4$ than the case with $H=1$. It should be noted that $H=1$ is not equivalent to ``no selection'', since any arc from and to the depot is still considered, and thus the shortest path can select any subset of consecutive customers. Thus the number of labels for $H=1$ is greater than one. Using the selection abilities of the algorithm, and thus values of $H$ greater than one, contributes to find solutions of higher quality in similar CPU time.

\paragraph{Sensitivity analysis on the auxiliary objective.} We also tested the impact of the auxiliary objective by comparing a version in which this objective is not considered, i.e. $\omega = 0$ to our base version. In these tests, and in the remainder of this paper, the value $H=3$ is used for the sparsification parameter. The results are presented in Table \ref{calib-obj} in the same format as previously.

\begin{table}[htb]
\caption{Calibration of the auxiliary objective}
\label{calib-obj}
\centering
\scalebox{0.88}
{
\begin{tabular}{|r|ll|}
\hline
&Auxiliary & Only Primary \\
&$\omega = 10^{-4}$ & $\omega = 0$ \\
\hline
Avg Gap&0.020\%&0.102\%\\
Best Gap&0.001\%&0.006\%\\
Nb BKS&155&150\\
Avg Time(s)&192.00&154.38\\
Avg T-Best(s)&91.09&56.43\\
Avg Labels&34.70&30.74\\
\hline
&\multicolumn{2}{c|}{\strut}\\
P-value & -- &  $5.1 \cdot 10^{-13}$  \\
\hline
\end{tabular}
}
\end{table}

Our results indicate a significant improvement, with a very small p-value of $5.1 \cdot 10^{-13}$, of the solution quality related to the use of the auxiliary objective during optimization. The auxiliary objective leads to deeper local optimums, and thus potentially longer descents. This is illustrated by an increase of $24.4\%$ of the average CPU time, and $61.4\%$ of the average time to find the best solution. Yet, as a consequence the solution quality is also much higher. In particular, the average gap to BKS is reduced from $0.102\%$ to $0.020\%$ when the auxiliary objective is used. Even by using longer runs or multiple runs without this objective, it would take much longer CPU time to reach such solution quality.

\subsection{Results on TOP}
\label{exp-TOP}

In this section, we compare the results generated by our methods to those generated by the previous state-of-the-art algorithms for the TOP, listed in Table \ref{BestMethod}. For this problem, several authors have reported results with very small CPU time, and thus to provide further elements of comparison we also created a \emph{fast version} of UHGS, named UHGS-f, by reducing the termination criteria to $(It_{max},  T_{max}) = (250, 60s)$. 

Tables \ref{results-TOP1}-\ref{results-TOP3} (in Appendix) provide detailed results on the 157 selected instances, in comparison to the best previous methods, as well as the new BKS, merging previously-known solutions and a few improved ones found by the proposed methods. When a method finds the BKS, this solution is highlighted in boldface. When a new improved BKS is found, this solution is also underlined. 

Table \ref{summary-TOP} provides a summary of these results. It displays for each method the average deviation to BKS on 10 runs, the average CPU time T(s) and average time T*(s) to reach the final solution  of a run, the deviation to BKS of the best solution out of 10 runs, the associated total CPU time, and the processor used in the tests. A last line reports the p-value of a Wilcoxon test on paired samples between the best solutions of UHGS and other methods. Detailed CPU times, on each instance, are available from the authors. Most of the previous authors only report their best solution out of ten runs. Thus the CPU time of these solutions, as well as the CPU time of our best solution out of ten runs, is scaled by a factor of ten.

\begin{table}[htbp]
\caption{Summary of results on the TOP}
\label{summary-TOP}
\hspace*{-1.55cm}
\scalebox{0.82}
{
\begin{tabular}{|r@{\hspace*{0.25cm}}l|cccccc|c|c|c|c|}
\hline
&& \textbf{FVF}& \textbf{SVF} & \textbf{ASe}& \textbf{FPR}& \textbf{SPR}& \textbf{MSA}& \textbf{UHGS}& \textbf{UHGS-f}& \textbf{MS-ILS}& \textbf{MS-LS}\\
\hline
 \textbf{Avg} &\textbf{Gap}&--&--&1.094\%&--&--&--& 0.020\%&0.100\%&0.115\%&0.861\%\\
 &\textbf{T(s)}&18.9&458&25.2&5.00&21.2&--&192&69.0&156&9.29\\
 &\textbf{T*(s)}&--&--&--&--&--&36.3&91.1&42.1&104&5.24\\
 \textbf{Best} & \textbf{Gap}&0.189\%&0.039\%&0.088\%&0.401\%&0.052\%&0.028\%&0.001\%&0.006\%&0.003\%&0.109\%\\
&\textbf{T(s)}&56.7&1370&252&50.0&212&--&1920&690&1560&92.9\\
 &\textbf{Nb$_\textsc{BKS}$}&94&134&128&78&126&133&155&150&153&115\\
 &\textbf{CPU}&P4 2.8G&P4 2.8G&PC 3G&Xe 2.5 G&Xe 2.5G&C2 2.5G&Xe 3.07G&Xe 3.07G&Xe 3.07G&Xe 3.07G\\
& \textbf{P-val}&5.2E-12&2.0E-05&3.9E-06&1.1E-14&1.2E-06&1.8E-05&***&4.3E-02&0.36&1.6E-08\\
\hline
\end{tabular}
}
\end{table}

The TOP results from methods in Table \ref{BestMethod} are displayed in Figure \ref{two-dim-TOP} on a two-dimensional plot considering CPU time and Gap(\%), \blue{using a logarithmic scale.} The DEC AlphaServer 1200/533 and DEC Alpha XP 1000 CPUs of \cite{Tang2005} are approximately $5\times$ to $10 \times$ slower than a Pentium IV 2.8 GHz according to Dongarra factors. The associated CPU time has thus been scaled accordingly.  Due to a lack of accurate data, the speed of the other processors can not be scaled in a reliable manner. These processors are of similar generation and should not differ by a speed factor of more than $2 \times$. Figure \ref{two-dim-TOP} thus displays a reasonable approximation of computation effort. Finally, it should be noted that \cite{Lin2013} reports the time to reach the best solution, which is smaller than the total computational time, and thus receives a significant advantage. \\

\blue{Overall, the proposed UHGS with large neighborhood search produces solutions of high quality, reaching a small gap of $0.001\%$. This is a significant improvement over previous methods, as  confirmed by the paired-samples Wilcoxon tests, with p-values smaller than $2 \cdot 10^{-5}$.
MS-ILS and MS-LS also obtain good solutions: the multi-start local improvement algorithm produces results of a quality comparable to current metaheuristics while being conceptually much simpler.
A total of 155/157 best known solutions have been found by UHGS, as well as three new BKS (1267 for p4.2.q, 1292 for p4.2.r, and 1120 for p7.3.t).}

\blue{The new methods with large neighborhoods, on the other hand, require a higher CPU time. For example, UHGS uses an average search time of 192 seconds, 91.1 seconds to find the final solution of the run, compared to 36.0 seconds to find the final solution for the previous best method, MSA. The CPU time of UHGS-f -- 42.1 seconds to find the solution -- is much closer to MSA. This increased time is due to the new neighborhoods, which require additional evaluation effort. To speed up the search, promising research avenues can consider the hybridization of the combined local search with more classic approaches, by applying CLS only on selected elite solutions.}

\begin{figure}[htbp]
\centering
\includegraphics[width=11cm]{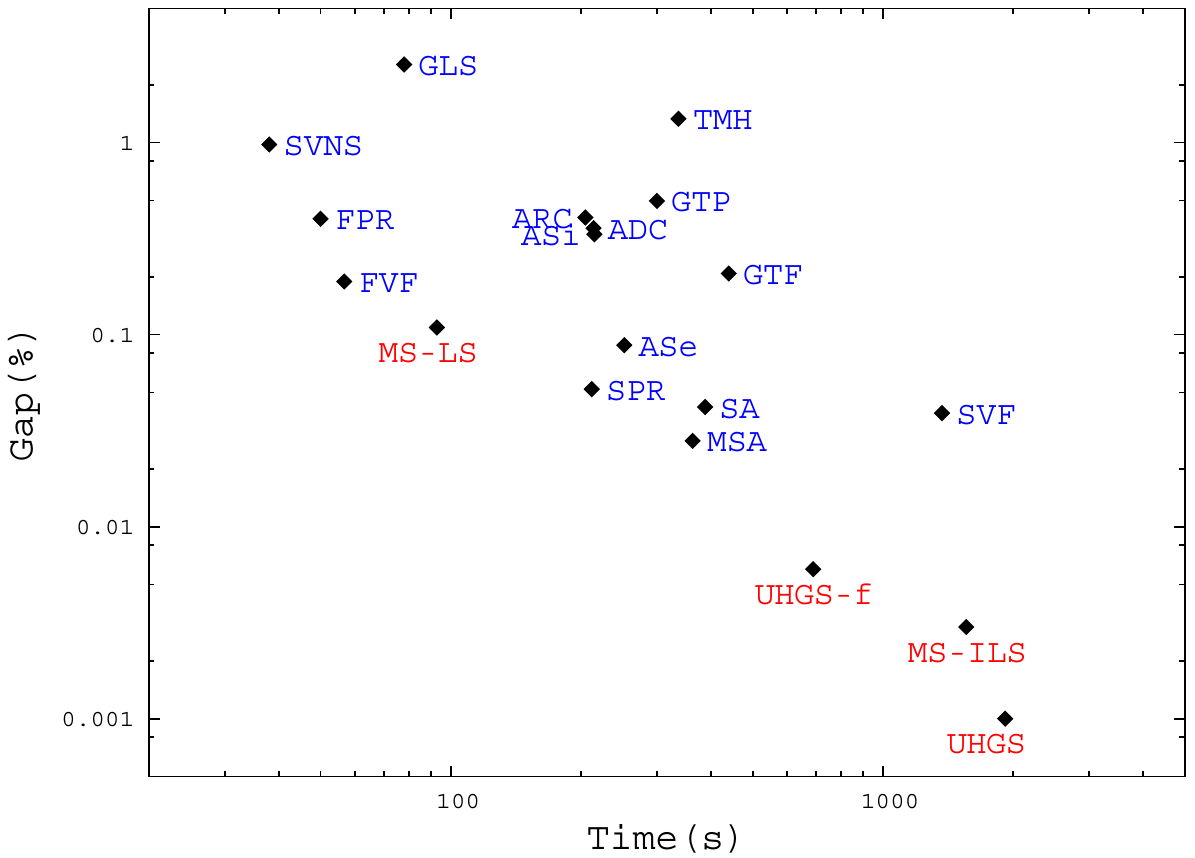}
\caption{Trade-off between solution quality and CPU time -- TOP benchmark instances}
\label{two-dim-TOP}
\end{figure}

\blue{Figure \ref{two-dim-TOP} highlights a wide range of alternative trade-offs between computational time and solution quality. Nine methods are in the pareto front (SVNS, FPR, FVF, SPR, MSA along with the four proposed methods) and thus, depending on the application context and the desired solution quality, a wide panel of methods with different levels of sophistication is available.}

\begin{figure}[!htbp]
\centering
\includegraphics[width=16cm]{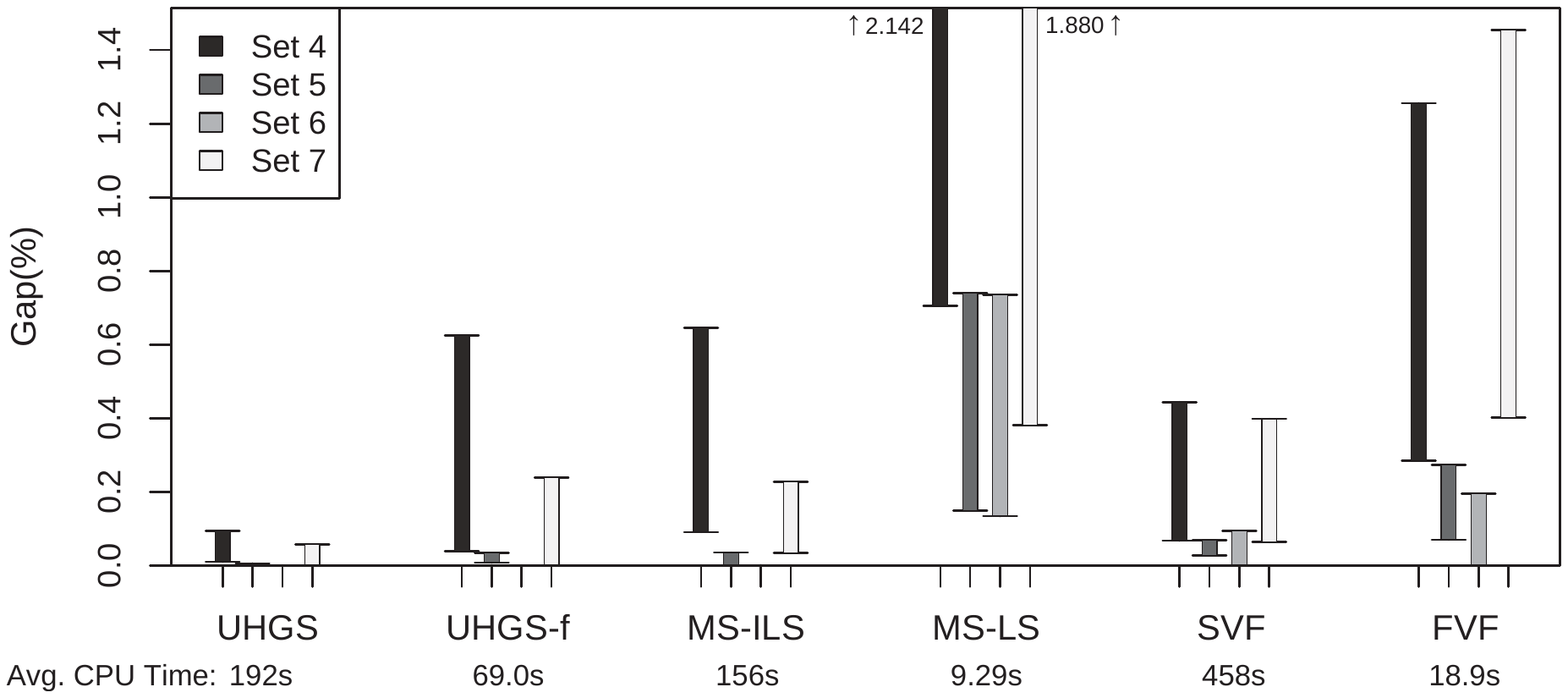}
\caption{\blue{Range of solutions for each method and each set of instances. The top of each interval represents the worst solution of three random runs, while the bottom represents the best solution out of three runs.}}
\label{two-dim-TOP}
\end{figure}

\blue{Finally, in a similar fashion as \cite{Archetti2006a}, Figure \ref{two-dim-TOP} displays an analysis of the worst and best solutions on three random runs. The data associated with this figure is provided in Table \ref{best-worst}, in appendix. The figure analyzes the robustness of methods, i.e., by how much solution quality may vary from one run to another. It appears that UHGS with large neighborhoods produces stable results (in 192 seconds per run) of overall higher quality than SVF (in 458 seconds per run) and other methods. A trade-off between computational effort and robustness is clearly observable. FVF appears is more stable than MS-LS, and requires less CPU time. Thus, the use of metaheuristic strategies still remains a good asset to achieve more~stable~results.}

\subsection{Results on other problems}

Further experiments have been conducted on the VRPPFCC and CPTP.
A summary of the results is displayed in Tables \ref{summary-VRPPFCC}-\ref{summary-CPTP} and Figure \ref{two-dim-others}, using the same conventions as previously. Average results, representative of a single run, are located with a triangle on this figure, while best solutions on ten runs are located with a diamond. The p-values in tables refer to statistical tests on the significance of the performance difference between the best solutions of UHGS and those of other methods. Detailed results per instance are provided in appendix.

\begin{table}[htb]
\caption{Summary of results on the VRPPFCC}
\label{summary-VRPPFCC}
\begin{minipage}{15.8cm}
\scalebox{0.8}
{
\begin{tabular}{|r@{\hspace*{0.25cm}}l|c@{\hspace*{0.3cm}}c@{\hspace*{0.3cm}}c@{\hspace*{0.3cm}}c@{\hspace*{0.3cm}}c|c|c|c|c|}
\hline
&&\textbf{RIP}&\textbf{TS}&\textbf{TS 2}&\textbf{TS+}&\textbf{AVNS}&\textbf{UHGS}&\textbf{MS-ILS}&\textbf{MS-LS}\\
\hline
\textbf{Avg}&\textbf{Gap CE}&1.104\%&0.408\%&--&--&--&0.153\%&0.331\%&2.862\%\\
&\textbf{Gap}&1.575\%&\textit{0.698\%}&--&--&--&0.272\%&0.682\%&3.155\%\\
&\textbf{T(min)}&17.45&5.77&14.21&34.86&11.94&26.39&16.61&1.89\\
\textbf{Best}&\textbf{Gap CE}&--&0.165\%&0.432\%&0.319\%&0.202\%&0.016\%&0.070\%&1.505\%\\
&\textbf{Gap}&--&\textit{0.405\%}&\textit{0.934\%}&\textit{0.228\%}& 0.422\%&0.038\%&0.270\%&2.100\%\\
&\textbf{T(min)}&--&57.74&142.11&348.59&119.44&263.95&166.15&18.88\\
&\textbf{Nb$_\textsc{BKS}$}&0&--&--&--&3&25&11&1\\
&\textbf{CPU}&Xe 3.6G&Opt 2.2G&Opt 2.2G&Opt 2.2G&I5 2.67G&Xe 3.07G&Xe 3.07G&Xe 3.07G\\
&\textbf{P-val}&6.5E-07&4.3E-02&2.1E-03&4.2E-03&1.2E-05&--&4.9E-04&5.4E-07\\
\hline
\end{tabular}
}
\vspace*{0.1cm}

\begin{footnotesize}
$^1$ \cite{Potvin2011} performed experiments with the same TS algorithm of \cite{Cote2009a}, yet the reported solution quality (TS 2) is lower than in the original paper. We display both set of results, but note that there may be a high variance of solution quality. 

$^2$ \cite{PrivCommVRPPFCC2013} reported to us that on Set ``G'', truncated distances have been erroneously used for CP09 and PN11. The related solutions and gaps, reported in italics, can thus be considered only as lower bounds.

\end{footnotesize}
\end{minipage}
\end{table}

\begin{table}[htb]
\caption{Summary of results on the CPTP}
\label{summary-CPTP}
\centering
\scalebox{0.9}
{
\begin{tabular}{|r@{\hspace*{0.25cm}}l|ccc|c|c|c|}
\hline
&&\textbf{VNS}&\textbf{GTF}&\textbf{GTP}&\textbf{UHGS}&\textbf{MS-ILS}&\textbf{MS-LS}\\
\hline
\textbf{Avg}&\textbf{Gap}&0.371\%&0.952\%&0.929\%&0.030\%&0.173\%&0.557\%\\
&\textbf{T(min)}&10.30&2.82&8.54&3.53&3.43&0.18\\
\textbf{Best}&\textbf{Gap}&--&--&--&0.001\%&0.015\%&0.154\%\\
&\textbf{T(min)}&--&--&--&35.28&34.31&1.85\\
&\textbf{Nb$_\textsc{BKS}$}&99&71&67&130&125&29\\
&\textbf{CPU}&PIV 2.8G&PIV 2.8G&PIV 2.8G&Xe 3.07G&Xe 3.07G&Xe 3.07G\\
&\textbf{P-val}&8.4E-06&1.4E-13&1.3E-13&--&4.3E-02&8.6E-02\\
\hline
\end{tabular}
}
\end{table}

For both problems, it is noteworthy that UHGS with the large neighborhood produces solution of better quality than the previous best methods. For the CPTP, a smaller CPU time is also achieved in most cases.
For the VRPPFCC, a gap to the BKS of $0.038\%$ is obtained, compared to $0.228\%$ for TS+, which was, in addition, advantaged because of a distance rounding issue (c.f. note in Table~\ref{summary-VRPPFCC}). The average solutions on one run of UHGS, with a gap of $0.272\%$, are also of higher quality than the best solutions of AVNS on ten runs, with a gap of $0.422\%$, while requiring less overall CPU time. UHGS produces a large number of best known solutions (25/34), including 20 new ones.
For the CPTP, UHGS reaches a total of 130/130 best known solutions, including 29 new ones, and all optimal solutions known from \cite{Archetti2008d} and \cite{Archetti2013a} have been retrieved.
MS-LS, MS-ILS and UHGS, lead to different CPU time and solution quality trade-offs, the fastest computation time being achieved with the simple MS-LS.  

\begin{figure}[htbp]
\centering
\hspace*{-1.0cm}
\includegraphics[width=8.3cm]{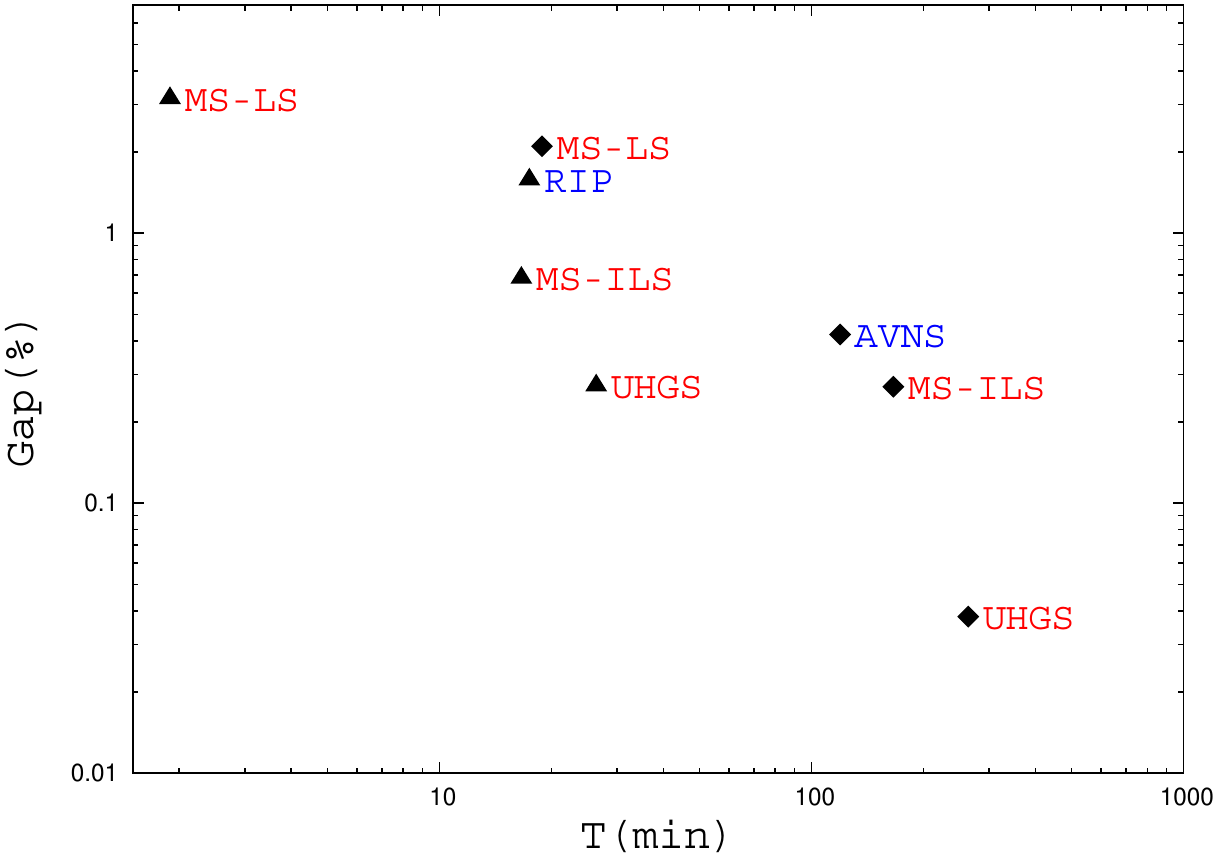}
\includegraphics[width=8.3cm]{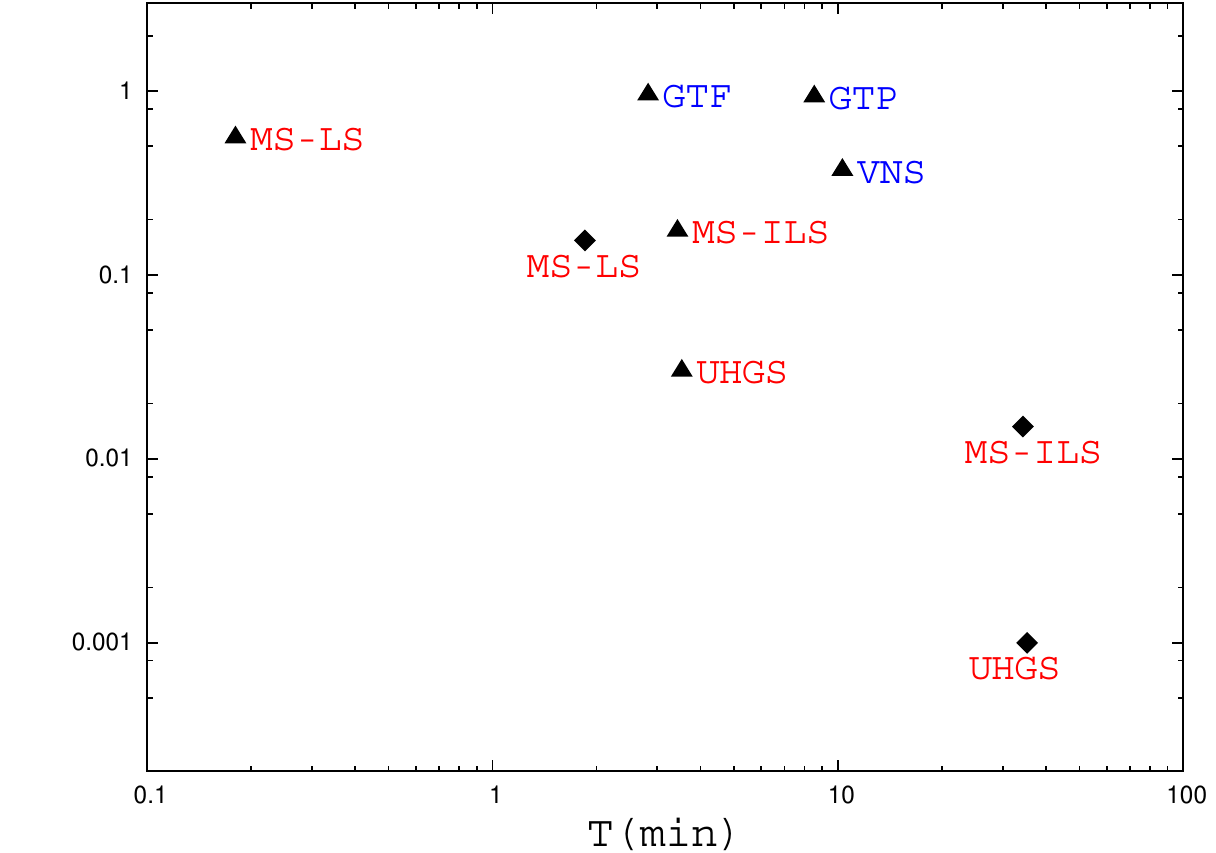}
\caption{Trade-off between solution quality and CPU time -- VRPPFCC (on the left) and CPTP (on the right). Comparison with state-of-the-art methods.}
\label{two-dim-others}
\end{figure}

On Figure \ref{two-dim-others}, the trade-off between solution quality and CPU time is visible for both problems. The proposed methods are dominating, since for any past method there exist at least one of our configurations which produces solutions of higher quality in less CPU time.
We also note that, while the stand-alone MS-LS turned out to be very efficient on the TOP, its performance on the VRPPFCC tends to be lower. This can be related to the fact that the objective of VRPPFCC or CPTP is based on both customer selection and routing. For the TOP, the selection alone is considered in the objective, and slightly sub-optimal routing decisions may only weakly impact the sets of feasible selections considered by the method. 

\section{Conclusion}

We have introduced a new large neighborhood for VRPs with profits.  These neighborhoods are searched by means of efficient pruning and bi-directional dynamic programming techniques. They have been tested in a local-improvement method, an iterated local search and a hybrid genetic search, on the TOP, the CPTP and the VRPPFCC. These new neighborhoods contribute to find solutions of higher quality in comparison to the previous state-of-the-art methods. 52 new best known solutions have been found. It is remarkable that the simple local-improvement method with these neighborhoods reaches solutions of similar quality than most current complex metaheuristics for the TOP.

The proposed method is a very novel way of designing neighborhood search on VRP with profits. It should be more successful on settings for which the ratio of customers to be delivered is high, such as the VRPPFCC where a small proportion of deliveries are usually assigned to a third party provider. Problems with a large number of delivery options, e.g. thousands of locations, and scarce resources, e.g. trucks to service only a few customers, may still not be well suited for such methodology, since all deliveries are considered in the routing algorithm. To efficiently handle  both cases, we suggest as a perspective of research to hybridize classic and new neighborhoods. The classic neighborhood can help filtering subsets of more promising deliveries, and generating elite initial solutions, which can then be improved by a few iterations of the large neighborhoods with subsets of potential customers. This would allow to harness the highest exploration capacities of our proposed neighborhood search while reducing even further the CPU time. Other perspective of research involve the extension of the proposed methodologies to other variants of VRPP, e.g., with time windows, variable profits, or arc routing, and even more general extensions of the concepts to other combinatorial optimization problems with decisions on task selections.


\begin{table}[!P]
\begin{footnotesize}
\vspace*{-2cm}
\renewcommand{\arraystretch}{0.95}
\caption{Results for the TOP, instances of \cite{Chao1996}}
\label{results-TOP1}
\hspace*{-0.7cm}
\scalebox{0.82}
{
\begin{tabular}{|c|ccccc|cc|cc|cc|cc|c|}
\hline
Inst&FVF&ASe&FPR&SPR&MSA&\multicolumn{2}{c|}{UHGS}&\multicolumn{2}{c|}{UHGS-f}&\multicolumn{2}{c|}{MS-ILS}&\multicolumn{2}{c|}{MS-LS}&BKS\\
&&&&&&Avg&Best&Avg&Best&Avg&Best&Avg&Best&\\
\hline
p4.2.a&\textbf{206}&\textbf{206}&\textbf{206}&\textbf{206}&\textbf{206}&\textbf{206}&\textbf{206}&\textbf{206}&\textbf{206}&\textbf{206}&\textbf{206}&\textbf{206}&\textbf{206}&206\\
p4.2.b&\textbf{341}&\textbf{341}&\textbf{341}&\textbf{341}&\textbf{341}&\textbf{341}&\textbf{341}&\textbf{341}&\textbf{341}&\textbf{341}&\textbf{341}&340&\textbf{341}&341\\
p4.2.c&\textbf{452}&\textbf{452}&\textbf{452}&\textbf{452}&\textbf{452}&\textbf{452}&\textbf{452}&\textbf{452}&\textbf{452}&\textbf{452}&\textbf{452}&450.8&\textbf{452}&452\\
p4.2.d&\textbf{531}&\textbf{531}&\textbf{531}&\textbf{531}&\textbf{531}&\textbf{531}&\textbf{531}&\textbf{531}&\textbf{531}&\textbf{531}&\textbf{531}&526.1&\textbf{531}&531\\
p4.2.e&\textbf{618}&\textbf{618}&612&\textbf{618}&\textbf{618}&\textbf{618}&\textbf{618}&617&\textbf{618}&\textbf{618}&\textbf{618}&603.7&613&618\\
p4.2.f&684&\textbf{687}&\textbf{687}&\textbf{687}&\textbf{687}&685.2&\textbf{687}&683.9&\textbf{687}&683.4&\textbf{687}&675.9&\textbf{687}&687\\
\hline
\end{tabular}
}
\end{footnotesize}
\end{table}

\begin{table}[!P]
\begin{footnotesize}
\vspace*{-2cm}
\renewcommand{\arraystretch}{0.95}
\caption{Results for the TOP, instances of \cite{Chao1996} (continued)}
\label{results-TOP2}
\hspace*{-0.7cm}
\scalebox{0.82}
{
\begin{tabular}{|c|ccccc|cc|cc|cc|cc|c|}
\hline
Inst&FVF&ASe&FPR&SPR&MSA&\multicolumn{2}{c|}{UHGS}&\multicolumn{2}{c|}{UHGS-f}&\multicolumn{2}{c|}{MS-ILS}&\multicolumn{2}{c|}{MS-LS}&BKS\\
&&&&&&Avg&Best&Avg&Best&Avg&Best&Avg&Best&\\
\hline
p5.2.z&1670&\textbf{1680}&1670&\textbf{1680}&\textbf{1680}&\textbf{1680}&\textbf{1680}&\textbf{1680}&\textbf{1680}&\textbf{1680}&\textbf{1680}&1668.5&\textbf{1680}&1680\\
p5.3.k&\textbf{495}&\textbf{495}&\textbf{495}&\textbf{495}&\textbf{495}&\textbf{495}&\textbf{495}&\textbf{495}&\textbf{495}&\textbf{495}&\textbf{495}&492&\textbf{495}&495\\
p5.3.l&\textbf{595}&\textbf{595}&\textbf{595}&\textbf{595}&\textbf{595}&\textbf{595}&\textbf{595}&\textbf{595}&\textbf{595}&\textbf{595}&\textbf{595}&591.5&\textbf{595}&595\\
p5.3.n&\textbf{755}&\textbf{755}&\textbf{755}&\textbf{755}&\textbf{755}&\textbf{755}&\textbf{755}&\textbf{755}&\textbf{755}&\textbf{755}&\textbf{755}&\textbf{755}&\textbf{755}&755\\
p5.3.o&\textbf{870}&\textbf{870}&\textbf{870}&\textbf{870}&\textbf{870}&\textbf{870}&\textbf{870}&\textbf{870}&\textbf{870}&\textbf{870}&\textbf{870}&\textbf{870}&\textbf{870}&870\\
p5.3.q&\textbf{1070}&\textbf{1070}&\textbf{1070}&\textbf{1070}&\textbf{1070}&1069.5&\textbf{1070}&\textbf{1070}&\textbf{1070}&\textbf{1070}&\textbf{1070}&1068&\textbf{1070}&1070\\
\hline
\end{tabular}
}
\end{footnotesize}
\end{table}

\begin{table}[!P]
\begin{footnotesize}
\renewcommand{\arraystretch}{0.95}
\caption{Results for the TOP, instances of \cite{Chao1996} (finished)}
\label{results-TOP3}
\hspace*{-0.2cm}
\scalebox{0.82}
{
\begin{tabular}{|c|ccccc|cc|cc|cc|cc|c|}
\hline
Inst&FVF&ASe&FPR&SPR&MSA&\multicolumn{2}{c|}{UHGS}&\multicolumn{2}{c|}{UHGS-f}&\multicolumn{2}{c|}{MS-ILS}&\multicolumn{2}{c|}{MS-LS}&BKS\\
&&&&&&Avg&Best&Avg&Best&Avg&Best&Avg&Best&\\
\hline
p7.3.r&1024&\textbf{1026}&1017&1021&\textbf{1026}&1025.8&\textbf{1026}&1024.1&\textbf{1026}&1022.9&\textbf{1026}&1016.1&1024&1026\\
p7.3.s&1079&\textbf{1081}&1076&\textbf{1081}&\textbf{1081}&1080.1&\textbf{1081}&1079.6&\textbf{1081}&1079.6&\textbf{1081}&1060.7&1074&1081\\
p7.3.t&1112&1118&1111&1118&1119&1119.2&\textbf{1120}&1117.9&\textbf{1120}&1116.9&\textbf{1120}&1105.7&1113&1120\\
p7.4.g&\textbf{217}&\textbf{217}&\textbf{217}&\textbf{217}&\textbf{217}&\textbf{217}&\textbf{217}&\textbf{217}&\textbf{217}&\textbf{217}&\textbf{217}&215.3&\textbf{217}&217\\
p7.4.h&\textbf{285}&\textbf{285}&\textbf{285}&\textbf{285}&\textbf{285}&\textbf{285}&\textbf{285}&\textbf{285}&\textbf{285}&\textbf{285}&\textbf{285}&284.4&\textbf{285}&285\\
p7.4.i&\textbf{366}&\textbf{366}&\textbf{366}&\textbf{366}&\textbf{366}&\textbf{366}&\textbf{366}&\textbf{366}&\textbf{366}&\textbf{366}&\textbf{366}&364.2&\textbf{366}&366\\
\hline
\end{tabular}
}
\end{footnotesize}
\end{table}

\begin{table}[htbp]
\centering
\caption{Best and Worst solutions on three runs, in Gap(\%) }
\label{best-worst}
\scalebox{0.9}
{
\begin{tabular}{|p{1.4cm}p{1.8cm}|p{1.5cm}p{1.5cm}p{1.5cm}p{1.5cm}p{1.5cm}p{1.5cm}|}
\hline
&&UHGS&UHGS-f&MS-ILS&MS-LS&SVF&FVF\\
\hline
\multirow{2}{*}{Set 4}&Best&0.009&0.038&0.090&0.704&0.067&0.284\\
&Worst&0.095&0.626&0.646&2.412&0.445&1.257\\
\hline
\multirow{2}{*}{Set 5}&Best&0.000&0.007&0.000&0.148&0.027&0.069\\
&Worst&0.007&0.036&0.036&0.740&0.070&0.274\\
\hline
\multirow{2}{*}{Set 6}&Best&0.000&0.000&0.000&0.133&0.000&0.000\\
&Worst&0.000&0.000&0.000&0.737&0.095&0.196\\
\hline
\multirow{2}{*}{Set 7}&Best&0.000&0.000&0.033&0.380&0.063&0.401\\
&Worst&0.058&0.240&0.228&1.880&0.399&1.456\\
\hline
\end{tabular}
}
\end{table}

\begin{table}[!P]
\begin{footnotesize}
\vspace*{-2cm}
\renewcommand{\arraystretch}{0.95}
\caption{Results for the VRPPFCC, instances of \cite{Bolduc2008}}
\label{results-VRPPFCC}
\hspace*{-2.3cm}
\scalebox{0.84}
{
\begin{minipage}{23.5cm}
\begin{tabular}{|cc|ccccc|c@{\hspace*{0.3cm}}c|c@{\hspace*{0.3cm}}c|c@{\hspace*{0.3cm}}c|c|}
\hline
Inst&n&RIP&TS&TS-PN&TS+&AVNS&\multicolumn{2}{c|}{UHGS}&\multicolumn{2}{c|}{MS-ILS}&\multicolumn{2}{c|}{MS-LS}&BKS\\
&&Single&Best&Best&Best&Best&Avg&Best&Avg&Best&Avg&Best&\\
\hline
p01&50&1132.91&\textbf{1119.47}&\textbf{1119.47}&\textbf{1119.47}&1123.95&1119.66&\textbf{1119.47}&1119.66&\textbf{1119.47}&1128.28&1121.32&1119.47\\
p02&75&1835.76&\textbf{1814.52}&\textbf{1814.52}&\textbf{1814.52}&\textbf{1814.52}&1815.63&\textbf{1814.52}&1817.34&\textbf{1814.52}&1883.44&1840.71&1814.52\\
p03&100&1959.65&1924.99&1921.10&1930.66&1920.86&1922.88&\textbf{1919.05}&1929.60&1922.18&1958.80&1943.64&1919.05\\
p04&150&2545.72&2515.50&2525.17&2525.17&2512.05&2509.82&\underline{\textbf{2505.39}}&2516.12&\underline{\textbf{2505.39}}&2568.49&2548.29&\underline{2505.39}\\
p05&199&3172.22&3097.99&3113.58&3117.10&3099.77&3095.58&\underline{\textbf{3081.59}}&3102.95&3090.53&3201.29&3181.86&\underline{3081.59}\\
p06&50&1208.33&\textbf{1207.47}&\textbf{1207.47}&\textbf{1207.47}&1207.81&\textbf{1207.47}&\textbf{1207.47}&1207.56&\textbf{1207.47}&1216.57&\textbf{1207.47}&1207.47\\
\hline
\end{tabular}

\vspace*{0.25cm}

\cite{PrivCommVRPPFCC2013} reported to us that on the set ``G'', truncated distances have been erroneously used for CP09 and PN11. The related solutions and gaps, reported in italics, can thus be considered only as lower bounds on the performance of these algorithms.

\end{minipage}
}
\end{footnotesize}
\end{table}

\begin{table}[!P]
\begin{footnotesize}
\renewcommand{\arraystretch}{0.95}
\vspace*{-2cm}
\caption{Results for the CPTP, instances of \cite{Archetti2008d}}
\label{results-CPTP}
\scalebox{0.82}
{
\begin{tabular}{|cc|ccc|c@{\hspace*{0.3cm}}c|c@{\hspace*{0.3cm}}c|c@{\hspace*{0.3cm}}c|c|}
\hline
Inst&n&VNS&GTF&GTP&\multicolumn{2}{c|}{UHGS}&\multicolumn{2}{c|}{MS-ILS}&\multicolumn{2}{c|}{MS-LS}&BKS\\
&&Single&Single&Single&Avg&Best&Avg&Best&Avg&Best&\\
\hline
p03-15-200&100&663.98&657.31&656.32&\underline{\textbf{664.92}}&\underline{\textbf{664.92}}&663.83&\underline{\textbf{664.92}}&657.84&663.78&\underline{664.92}\\
p03-2-50&100&\textbf{57.75}&\textbf{57.75}&\textbf{57.75}&\textbf{57.75}&\textbf{57.75}&\textbf{57.75}&\textbf{57.75}&\textbf{57.75}&\textbf{57.75}&57.75\\
p03-2-75&100&\textbf{106.15}&\textbf{106.15}&\textbf{106.15}&\textbf{106.15}&\textbf{106.15}&\textbf{106.15}&\textbf{106.15}&105.98&\textbf{106.15}&106.15\\
p03-2-100&100&\textbf{158.21}&\textbf{158.21}&\textbf{158.21}&\textbf{158.21}&\textbf{158.21}&\textbf{158.21}&\textbf{158.21}&158.17&\textbf{158.21}&158.21\\
p03-2-200&100&\textbf{330.14}&319.28&319.28&\textbf{330.14}&\textbf{330.14}&\textbf{330.14}&\textbf{330.14}&327.34&\textbf{330.14}&330.14\\
p03-3-50&100&\textbf{80.82}&\textbf{80.82}&\textbf{80.82}&\textbf{80.82}&\textbf{80.82}&\textbf{80.82}&\textbf{80.82}&\textbf{80.82}&\textbf{80.82}&80.82\\
\hline
\end{tabular}
}
\end{footnotesize}
\end{table}

\begin{table}[!P]
\begin{footnotesize}
\renewcommand{\arraystretch}{0.95}
\vspace*{-2cm}
\caption{Results for the CPTP, instances of \cite{Archetti2008d} (continued)}
\label{results-CPTP}
\scalebox{0.82}
{
\begin{tabular}{|cc|ccc|c@{\hspace*{0.3cm}}c|c@{\hspace*{0.3cm}}c|c@{\hspace*{0.3cm}}c|c|}
\hline
Inst&n&VNS&GTF&GTP&\multicolumn{2}{c|}{UHGS}&\multicolumn{2}{c|}{MS-ILS}&\multicolumn{2}{c|}{MS-LS}&BKS\\
&&Single&Single&Single&Avg&Best&Avg&Best&Avg&Best&\\
\hline
p10-20-200&199&1773.65&1761.37&1759.81&1788.6&\underline{\textbf{1793.95}}&1776.49&1785&1725.81&1756.56&\underline{1793.95}\\
p10-2-50&199&\textbf{70.87}&\textbf{70.87}&\textbf{70.87}&\textbf{70.87}&\textbf{70.87}&\textbf{70.87}&\textbf{70.87}&\textbf{70.87}&\textbf{70.87}&70.87\\
p10-2-75&199&\textbf{124.85}&\textbf{124.85}&\textbf{124.85}&\textbf{124.85}&\textbf{124.85}&\textbf{124.85}&\textbf{124.85}&\textbf{124.85}&\textbf{124.85}&124.85\\
p10-2-100&199&\textbf{171.24}&\textbf{171.24}&\textbf{171.24}&\textbf{171.24}&\textbf{171.24}&\textbf{171.24}&\textbf{171.24}&\textbf{171.24}&\textbf{171.24}&171.24\\
p10-2-200&199&\textbf{382.41}&378.32&379.81&382&\textbf{382.41}&382.03&\textbf{382.41}&380.71&\textbf{382.41}&382.41\\
p10-3-50&199&\textbf{103.79}&\textbf{103.79}&\textbf{103.79}&\textbf{103.79}&\textbf{103.79}&\textbf{103.79}&\textbf{103.79}&\textbf{103.79}&\textbf{103.79}&103.79\\
\hline
\end{tabular}
}
\end{footnotesize}
\end{table}

\end{document}